\documentclass[11pt]{article}

\usepackage{amsmath,amssymb,amsthm}
\usepackage{graphicx}

\usepackage{hyperref}

\usepackage[letterpaper, margin=1in]{geometry}

\usepackage{cite}

\usepackage[toc,page]{appendix}

\makeatletter
\newtheorem*{rep@theorem}{\rep@title}
\newcommand{\newreptheorem}[2]{%
	\newenvironment{rep#1}[1]{%
		\def\rep@title{#2 \ref{##1}}%
		\begin{rep@theorem}}%
		{\end{rep@theorem}}}
\makeatother

\newenvironment{lemma-repeat}[1]{\begin{trivlist}
		\item[\hspace{\labelsep}{\bf\noindent Lemma \ref{#1} }]\em }%
	{\end{trivlist}}
\newenvironment{theorem-repeat}[1]{\begin{trivlist}
		\item[\hspace{\labelsep}{\bf\noindent Theorem \ref{#1} }]\em }%
	{\end{trivlist}}

\newcommand*\samethanks[1][\value{footnote}]{\footnotemark[#1]}

\newtheorem{theorem}{Theorem}[section]
\newtheorem{lemma}[theorem]{Lemma}
\newtheorem{corollary}[theorem]{Corollary}
\newtheorem{claim}[theorem]{Claim}
\newtheorem{definition}[theorem]{Definition}

\newtheorem{observation}[theorem]{Observation}

\newtheorem{remark}[theorem]{Remark}

\def\eps{\varepsilon}
\def\poly{poly}

\begin{document}

\newcommand{\ThmSpa}
{
	The number of rounds needed for any protocol to compute the diameter of a network on $n$ nodes and $O(n \log n)$ edges of constant diameter in the $CONGEST$ model is $\Omega(\frac{n}{\log^2{n}})$.
}

\newcommand{\ThmDA}
{
	For all constant $0<\eps<1/2$, the number of rounds needed for any protocol to compute a $(3/2-\varepsilon)$-approximation to the diameter of a sparse network is $\Omega( \frac{ n}{\log^3n})$.
}

\newcommand{\ThmRadius}
{
	For all constant $0<\eps<1/2$, the number of rounds needed for any protocol to compute a $(3/2-\varepsilon)$-approximation to the radius of a sparse network is $\Omega(\frac{ n}{\log^3n})$.
}

\newcommand{\ThmEcc}
{
	For all constant $0<\eps<2/3$, the number of rounds needed for any protocol to compute a $(5/3-\varepsilon)$-approximation of all eccentricities of a sparse network is $\Omega(\frac{ n}{\log^3n})$.
}
\newcommand{\ThmSpanners}
{
	Given an \emph{unweighted} graph $G=V,E$ and a subgraph $H\subset E$ of $G$, the number of rounds needed for any protocol to decide whether $H$ is an $(\alpha,\beta)$-spanner of $G$ in the $CONGEST$ model is $\Omega(\frac{n}{(\alpha+\beta)\log^3{n}})$, for any $\alpha<\beta+1$.
}
\newcommand{\ThmDeg}
{
	The number of rounds needed for any protocol to compute the radius of a sparse network of constant degree in the $CONGEST$ model is $\Omega(n/\log^3{n})$.
}

\begin{titlepage}
	\title{Near-Linear Lower Bounds for Distributed Distance Computations, Even in Sparse Networks}
	\author{Amir Abboud\thanks{Stanford University, Department of Computer Science, \texttt{abboud@cs.stanford.edu}. Supported by Virginia Vassilevska Williams's NSF Grants CCF-1417238 and CCF-1514339, and BSF Grant BSF:2012338.}
		\and Keren Censor-Hillel\thanks{Technion, Department of Computer Science, \texttt{\{ckeren,serikhoury\}@cs.technion.ac.il}. Supported by ISF Individual Research Grant 1696/14.} \and Seri Khoury\samethanks}
	\maketitle
	
	\begin{abstract}

\iffalse
		In this work we present stronger lower bounds on the number of rounds needed to compute distance graph parameters such as diameter and radius in the $CONGEST$ model of distributed computing. Here, in each round, each node of the network can transmit a short message to each of its neighbors. Frishknecht et al.~\cite{FrischknechtHW12} showed that $\widetilde{\Omega}(n)$ rounds are needed for any algorithm to compute the diameter of a dense network in the $CONGEST$ model, i.e., a network of $n$ nodes and $\Theta(n^2)$ edges. Our first contribution is $\widetilde{\Omega}(n)$ lower bound on the number of rounds needed to compute the diameter of a sparse network as well. Our technique allows us to improve previous results, for instance, we improve the $\widetilde{\Omega}(\sqrt{n})$ lower bound presented by~\cite{FrischknechtHW12} for any algorithm to compute a $(\frac{3}{2}-\varepsilon)$-approximation to the diameter by showing an $\widetilde{\Omega}(n)$ lower bound for the same approximation factor. Furthermore, we present, to the best of our knowledge, the first non-trivial lower bound for computing the radius of a network, and improve the $\widetilde{\Omega}(\sqrt{n})$ lower bound reported by Holzer and wattenhofer~\cite{HolzerW12} on the number of rounds needed to compute a $(\frac{3}{2}-\varepsilon)$-approximation to all eccentricities by showing an $\widetilde{\Omega}(n)$ lower bound on the number of rounds needed even to compute $(\frac{5}{3}-\varepsilon)$-approximation to all eccentricities.
\fi

We develop a new technique for constructing sparse graphs that allow us to prove near-linear lower bounds on the round complexity of computing distances in the CONGEST model.
Specifically, we show an $\widetilde{\Omega}(n)$ lower bound for computing the diameter in sparse networks, which was previously known only for dense networks [Frishknecht et al., SODA 2012]. In fact, we can even modify our construction to obtain graphs with constant degree, using a simple but powerful degree-reduction technique which we define.

Moreover, our technique allows us to show $\widetilde{\Omega}(n)$ lower bounds for computing $(\frac{3}{2}-\varepsilon)$-approximations of the diameter or the radius, and for computing a $(\frac{5}{3}-\varepsilon)$-approximation of all eccentricities. For radius, we are unaware of any previous lower bounds. For diameter, these greatly improve upon previous lower bounds and are tight up to polylogarithmic factors [Frishknecht et al., SODA 2012], and for eccentricities the improvement is both in the lower bound and in the approximation factor [Holzer and Wattenhofer, PODC 2012].

Interestingly, our technique also allows showing an almost-linear lower bound for the verification of $(\alpha,\beta)$-spanners, for $\alpha < \beta+1$.

	\end{abstract}

\thispagestyle{empty}
\end{titlepage}

\section{Introduction}

The diameter and radius are two basic graph parameters whose values play a vital role in many applications. In distributed computing, these parameters are even more fundamental, since they capture the minimal number of rounds needed in order to send a piece of information to all the nodes in a network. Hence, understanding the complexity of computing these parameters is central to distributed computing, and has been the focus of many studies in the CONGEST model of computation, where in every round each of $n$ nodes may send messages of up to $O(\log{n})$ bits to each of its neighbors. Frischknecht et al.~\cite{FrischknechtHW12} showed that the diameter is surprisingly hard to compute: $\widetilde{\Omega}(n)$ rounds are needed even in networks with constant diameter.\footnote{The notations $\widetilde{\Omega}$ and $\widetilde{O}$ hide factors that are polylogarithmic in $n$. } This lower bound is nearly tight, due to an $O(n)$ upper bound presented by~\cite{PelegRT12} to compute all pairs shortest paths in a network.
Naturally, approximate solutions are a desired relaxation, and were indeed addressed in several cornerstone studies~\cite{HolzerPRW14,PelegRT12,HolzerW12,LenzenP13,FrischknechtHW12}, bringing us even closer to a satisfactory understanding of the time complexity of computing distances in distributed networks. However, several central questions remained elusive.

\paragraph{Sparse Graphs.} The graphs constructed in~\cite{FrischknechtHW12} have $\Theta(n^2)$ edges and constant diameter, and require any distributed protocol for computing their diameter to spend $\widetilde{\Omega}(n)$ rounds.
Such a high lower bound makes one wonder if the diameter can be computed faster in networks that we expect to encounter in realistic applications.
Almost all large networks of practical interest are very sparse~\cite{SNAP14}, e.g. the Internet in 2012 had $\approx 4$ billions nodes and $\approx 128$ billion edges~\cite{MeuselVLB15}.

The only known lower bound for computing the diameter of a sparse network is obtained by a simple modification to the construction of~\cite{FrischknechtHW12} which yields a much weaker bound of $\widetilde{\Omega}(\sqrt{n})$.
This leaves hope that the $\widetilde{\Omega}(n)$ bound can be beaten significantly in sparse networks.
Our first result is to rule out this possibility.

\begin{theorem}
	\label{thm:sparse}
	\ThmSpa
\end{theorem}

We remark that, as in~\cite{FrischknechtHW12}, our lower bound holds even for networks with constant diameter and even against randomized algorithms.
Throughout the paper we say that a graph on $n$ nodes is sparse if it has $O(n \log{n})$ edges.
Due to simple transformations, e.g. by adding dummy nodes, all of our lower bounds will also hold for the more strict definition of sparse graphs as having $O(n)$ edges, up to a loss of a log factor.

As explained next, the sparsity in our new lower bound construction allows us to extend the result in some interesting ways.

\paragraph{Approximation Algorithms.}
Another important question is whether we can bypass this near-linear barrier if we settle for knowing only an approximation to the diameter.
An $\alpha$-approximation algorithm to the diameter returns a value $\hat{D}$ such that $D \leq \hat{D} \leq \alpha \cdot D$, where $D$ is the diameter of the network.
From~\cite{FrischknechtHW12} we know that $\widetilde{\Omega}(\sqrt{n}+D)$ rounds are needed, even for computing a $(\frac{3}{2}-\varepsilon)$-approximation to the diameter, for any constant $\eps>0$.

Following this lower bound, almost-complementary upper bounds were under extensive research.
It is known that a $\frac{3}{2}$-approximation can be computed in a sublinear number of rounds: Holzer and Wattenhofer~\cite{HolzerW12} showed a $O(n^{3/4}+D)$-round algorithm and (independently) Peleg et al.~\cite{PelegRT12} obtained a $O(D \sqrt{n} \log{n})$ bound, later these bounds were improved to $O(\sqrt{n}\log{n}+D)$ by Lenzen and Peleg~\cite{LenzenP13}, and finally Holzer et al.~\cite{HolzerPRW14} reduce the bound to $O(\sqrt{n\log{n}}+D)$. When $D$ is small, these upper bounds are near-optimal in terms of the round complexity -- but do they have the best possible approximation ratio that can be achieved within a sublinear number of rounds? That is, can we also obtain a $(\frac{3}{2}-\varepsilon)$-approximation in $\widetilde{O}(\sqrt{n}+D)$ rounds, to match the lower bound of~\cite{FrischknechtHW12}?

Progress towards answering this question was made by Holzer and Wattenhofer~\cite{HolzerW12} who showed that any algorithm that needs to decide whether the diameter is $2$ or $3$ has to spend $\widetilde{\Omega}(n)$ rounds.
However, as the authors point out, their lower bound is not robust and does not rule out the possibility of a $(\frac{3}{2}-\varepsilon)$-approximation when the diameter is larger than $2$, or an algorithm that is allowed an additive $+1$ error besides a multiplicative $(\frac{3}{2}-\varepsilon)$ error.

Perhaps the main difficulty in extending the lower bound constructions of~\cite{FrischknechtHW12} and~\cite{HolzerW12} to resolve these gaps was that their original graphs are dense.
A natural way to go from a lower bound construction against exact algorithms to a lower bound against approximations is to subdivide each edge into a path.
When applied to dense graphs, this transformation blows up the number of nodes quadratically, resulting in a $\widetilde{\Omega}(\sqrt{n})$ lower bound~\cite{FrischknechtHW12}.
 Our new sparse construction technique allows us to tighten the bounds and negatively resolve the above question. In particular, we show a $\widetilde{\Omega}(n)$ lower bound for computing a ($\frac{3}{2}-\varepsilon$)-approximation to the diameter.

\begin{theorem}
	\label{thm:DA}
	\ThmDA
\end{theorem}

\paragraph{Radius.}
In many scenarios we want one special node to be able to efficiently send information to all other nodes. In this case, we would like this node to be the one that is closest to every other node, i.e. the \emph{center} of the graph.
The \emph{radius} of the graph is the largest distance from the center, and it captures the number of rounds needed for the center node to transfer a message to another node in the network.
While radius and diameter are closely related, the previous lower bounds for diameter do not transfer to radius and it was conceivable that the radius of the graph could be computed much faster.
Obtaining a non-trivial lower bound for radius has been stated as an open problem in~\cite{HolzerW12}. A third advantage of our technique is that it extends to computing the radius, for which we show that the same strong near-linear barriers above hold.

	\begin{theorem}
		\label{thm:Radius}
		\ThmRadius
	\end{theorem}

\paragraph{Eccentricity}
The eccentricity of a node is the largest distance from it. Observe that the diameter is the largest eccentricity in the graph while the radius is the smallest.
As pointed in~\cite{HolzerW12}, given a $(\frac{3}{2}-\varepsilon)$-approximation algorithm to all the eccentricities, we can achieve $(\frac{3}{2}-\varepsilon)$-approximation algorithm to the diameter by a simple flooding. This implies an $\widetilde{\Omega}(\sqrt{n}+D)$ lower bound for any $(\frac{3}{2}-\varepsilon)$-approximation algorithm for computing all the eccentricities. Our construction allows us to improve this result by showing that any algorithm for computing even a $(\frac{5}{3}-\varepsilon)$-approximation to all the eccentricities must spend $\Omega(\frac{n}{\log^3(n)})$ rounds.
This improves both in terms of the number of rounds, and in terms of the approximation factor, which we allow to be even larger.
Interestingly, it implies that approximating all eccentricities is even harder than approximating just the largest or the smallest one.

	\begin{theorem}
		\label{thm:ecc}
		\ThmEcc
	\end{theorem}

\paragraph{Constant-Degree Graphs} For computing exact diameter and radius, we can modify the graph constructions according to a \emph{degree-reduction technique} we define, such that the resulting graphs have a constant degree, and still allow us to obtain near-linear lower bounds. Roughly speaking, given a node $v$, we replace a subset of $y$ edges of $v$ by a binary tree to the respective neighbors (with additional internal nodes). This reduces the degree of $v$ by $y-2$. Repeatedly applying this procedure in a careful manner results in a graph of constant degree, for which we can still show our near-linear lower bounds (notice that distances change, as well as the number of nodes). We exemplify this technique by obtaining the following lower bound on computing the radius.

\begin{theorem}
		\label{ExactRadsmallDeg}
		\ThmDeg
	\end{theorem}

\paragraph{Verification of Spanners} Finally, our technique allows us to obtain a lower bound for the verification of $(\alpha,\beta)$-spanners. An $(\alpha,\beta)$-spanner of a graph $G$, is a subgraph $H$ in which for any two nodes $u,v$ it holds that $d_H(u,v) \leq \alpha d_G(u,v)+\beta$. When spanners are sparse, i.e., when $H$ does not have too many edges, they play a vital role in many application, such as routing, approximating distances, synchronization, and more. Hence, the construction of sparse spanners has been a central topic of many studies, both in centralized and sequential computing.

Here we address the problem of verifying that a given subgraph $H$ is indeed an $(\alpha,\beta)$-spanner of $G$. At the end of the computation, each node outputs a bit indicating whether $H$ is a spanner, with the requirement that if $H$ is indeed a spanner with the required parameters then all nodes indicate this, and if it is not then at least one node indicates that it is not. We obtain the following.

	\begin{theorem}\label{thm:spanners}
		\ThmSpanners
	\end{theorem}

Notice that for any reasonable value of $\alpha,\beta = O(\poly\log{n})$, the lower bound is near-linear. This is another evidence for a task for which verification is harder than computation in the CONGEST model, as initially brought into light in~\cite{SarmaHKKNPPW12}. This is, for example, because $(+2)$-purely additive spanners with $O(n^{3/2}\log{n})$ edges can be constructed in $O(\sqrt{n}\log{n} + D)$ rounds (this appears in~\cite{LenzenP13}, and can also be deduced from~\cite{HolzerW12}), and additional various additive spanners can be constructed fast in CONGEST~\cite{CKTY16}.

\subsection{Techniques}

\paragraph{Communication Complexity and Distributed Computing.}A well-known technique to
prove lower bounds in the $CONGEST$ model is to use a reduction from
communication complexity to distributed computing. Peleg and Rubinovich\cite{PelegR99} apply a
lower bound from communication complexity to show that the number of rounds needed for any
distributed algorithm to construct a minimum spanning tree (MST) is $\widetilde{\Omega}(\sqrt{n}+D)$. Many recent papers were inspired by this technique. In\cite{Elkin04} Elkin extended the result of~\cite{PelegR99} to
show that any distributed algorithm for constructing an $\alpha$-approximation to the MST must spend
$\widetilde{\Omega}(\sqrt{\frac{n}{\alpha}})$ rounds. Das Sarma et al.\cite{SarmaHKKNPPW12} show that
any distributed verification algorithm for many problems, such as connectivity, $s-t$ cut and
approximating MST requires $\widetilde{\Omega}(\sqrt{n}+D)$ rounds. Nanongkai et al.
\cite{NanongkaiSP11} showed an $\Omega(\sqrt{\ell \cdot D}+D)$ lower bound for
computing a random walk of length $\ell$. Similar reductions from communication
complexity were adapted also in the $CONGEST\ Clique\ Broadcast$ model
\cite{DruckerKO13,HolzerP14}, where in each round each node can broadcast the same
$O(\log{n})$-bit message to all the nodes in the network.

Similar to the technique used in~\cite{FrischknechtHW12,DruckerKO13,SarmaHKKNPPW12,HolzerW12,HolzerP14}, our lower bounds are obtained by reductions from the Set-Disjointness problem in the two-party number-in-hand model of communication complexity
\cite{Yao79}. Here, each of the players Alice and Bob receives a $k$-bit string, $S_a$ and $S_b$
respectively, and needs to decide whether the two strings are disjoint or not, i.e., whether there is
some bit $0\leq i\leq k-1$ such that $S_a[i]=1$ and $S_b[i]=1$. A classical result~\cite{Razborov92,Kushilevitz} is that in
order to solve the Set-Disjointness problem, Alice and Bob must exchange $\Omega(k)$ bits.

The high level idea for applying this lower bound in the $CONGEST$ model, is as follows.
We construct a graph in which the existence of some of the edges depends on the inputs of Alice and Bob, and we partition the graph between the two players, inducing a cut in it, which we will refer to as the ``communication-cut".
The graph will have some property (e.g. diameter at least 4) if and only if the two strings of Alice and Bob are disjoint.
The players can then simulate a distributed algorithm (e.g. for diameter), while exchanging only the bits that are sent by the algorithm on edges that belong to the communication-cut.
If our cut has $t$ edges, then Alice on Bob only exchange $O(r \cdot t \cdot \log{n})$ bits where $r$ is an upper bound on the round complexity of the algorithm.
Therefore, the lower bound on the communication complexity of Set-Disjointness implies a lower bound on the number of \emph{rounds} required for any distributed algorithm (for diameter). Observe that the larger the communication-cut in the reduction, the smaller the lower bound for the distributed problem.

Having \emph{a sparse graph with a small cut}, is what allows us to make this leap in the lower bounds. To achieve this, the key idea is to connect the nodes to a set of nodes that represent their binary value, and the only nodes on the cut are the nodes of the binary representation. We call this graph structure a \emph{bit-gadget}, and it plays a central role in all of our graph constructions. This is inspired by graph constructions for different settings (e.g.~\cite{AbboudGW15}, see additional discussion for sequential algorithms below).

\subsection{Additional Related Work}

There are many known upper~\cite{LenzenPS13,Nanongkai14,HenzingerKN15} and lower~\cite{Elkin06,PelegR99,KorKP13,SarmaHKKNPPW12} bounds for approximate distance computation in \emph{weighted} networks.
For example, the weighted diameter of a network with underlying diameter $D$ can be approximated to within $(2+o(1))$ in $O(n^{1/2+o(1)}+D^{1+o(1)})$ rounds~\cite{HenzingerKN15}.
Moreover, such problems have also been considered in the \emph{congested clique} model~\cite{Nanongkai14,HenzingerKN15,CensorKKLPS15}, where $(1+o(1))$-approximate all pairs shortest paths can be computed in $O(n^{0.158})$ rounds~\cite{CensorKKLPS15}.

\paragraph{Diameter and Radius in Sequential Algorithms.}
Intuitively, the technical difficulty in extending the proof for diameter to work for radius as well is the difference in types between the two problems: the diameter asks for a pair of nodes that are far ($\exists x \exists y$) while radius asks for a node that is close to everyone ($\exists x \forall y$).
Recent developments in the theory of (sequential) algorithms suggest that this type-mismatch could lead to fundamental differences between the two problems. Recall that classical \emph{sequential} algorithms solve APSP in $O(nm)$ time~\cite{CLRS} and therefore both diameter and radius can be solved in quadratic $O(n^2)$ time in sparse graphs.

Due to the lack of techniques for proving \emph{unconditional} super-linear $\omega(n)$ lower bounds on the runtime of sequential algorithms for any natural problem, a recent line of work seeks hardness results conditioned on certain plausible conjectures (a.k.a. ``Hardness in P").
An interesting example of such result concerns the diameter:
Roditty and Vassilevska W.~\cite{RodittyW13} proved that if the diameter of sparse graphs can be computed in truly-subquadratic $O(n^{2-\eps})$ time, for any $\eps>0$, then the Strong Exponential Time Hypothesis (SETH) is false\footnote{SETH is a pessimistic version of the ${\sf P} \neq {\sf NP}$ conjecture, which essentially says that CNF-SAT cannot be solved in $(2-\eps)^n$ time. More formally, SETH is the assumption that there is no $\eps>0$ such that for all $k \geq 1$ we can solve $k$-SAT on $n$ variables and $m$ clauses in $(2-\eps)^n \cdot poly(m)$ time.}, by reducing SAT to diameter.
Since then, many other problems were shown to be ``SETH-hard" (e.g.~\cite{AbboudWW14,AbboudW14,AbboudWY15,BackursI15,AbboudBHWZ16} to name a few) but whether a similar lower bound holds for radius is an open question~\cite{RodittyW13,ChechikLRSTW14,AbboudGW15,BorassiCH14,AbboudWW16,CairoGR16}.
In fact, Carmosino et al.~\cite{CarmosinoGIMPS16} show that there is a formal barrier for reducing SAT to radius\footnote{It would imply a new co-nondeterministic algorithm for SAT and refute the Nondeterministic-SETH, which is a strong version of ${\sf NP} \neq {\sf CONP}$.}, and Abboud, Vassilevska W. and Wang~\cite{AbboudWW16}  introduce a new conjecture to prove an $n^{2-o(1)}$ lower bound for radius\footnote{A truly-subquadratic algorithm for computing the radius of a sparse graph refutes the Hitting Set Conjecture: there is no $\eps>0$ such that given two lists $A,B$ of $n$ subsets of a universe $U$ of size $\poly\log{n}$ we can decide whether there is a set $a \in A$ that intersects all sets $b \in B$ in $O(n^{2-\eps})$ time.} (which has a similar $\exists \forall$ type).
Diameter and radius seem to behave differently also in the regime of dense and weighted graphs where the best known algorithms take roughly cubic $O(n^{3}/2^{\sqrt{\log{n}}})$ time~\cite{Williams14,ChanW16} and it is known that radius can be solved in truly-subcubic $O(n^{3-\eps})$ time if and only if APSP can~\cite{AbboudGW15}, but showing such a subcubic-equivalence between APSP and diameter is a big open question~\cite{AingworthCIM99,WilliamsW10,AbboudGW15}.

The framework and set-up in our \emph{unconditional} lower bound proofs for distributed algorithms are very different from the ones in the works on \emph{conditional} lower bounds for sequential algorithms discussed above.
Still, some of our graph gadgets are inspired by the constructions in those proofs, e.g.~\cite{RodittyW13,ChechikLRSTW14,AbboudGW15,AbboudWW16,CairoGR16}.
Thus, it is quite surprising that our hardness proof for diameter transfers without much difficulty to a hardness proof for radius.

\subsection{Model and basic definitions}

We consider a synchronized network of $n$ nodes represented by an undirected graph $G=(V,E)$.
In each round, each node can send a different message of $b$ bits to each of its neighbors. This model is known as the $CONGEST(b)$ model, and as the $CONGEST$ model when $b=O(\log(n))$~\cite{congest}.

The graph parameters that need to be computed are formally defined as follows.

\begin{definition}(Eccentricity, Diameter and Radius)
Let $d(u,v)$ denote the length of the shortest path between the nodes $u$ and $v$.
The eccentricity $e(u)$ of some node $u$ is $max_{v\in V}d(u,v)$. The Diameter (denoted by $D$) is the maximum distance between any two nodes in the graph: $D = max_{u\in V}e(u)$. The Radius (denoted by $r$) is the maximum distance from some node to the ``center" of the graph: $r$ = $min_{u\in V}e(u)$.
\end{definition}
Finally, we define what we mean when we say that a graph is sparse.
\begin{definition}(sparse network)
	A sparse network $G=(V,E)$ is a network with $n$ nodes and at most $O(n\log(n))$ edges.
\end{definition}

Recall, however, that all our results can be obtained for graphs that have a strictly linear number of $\Theta(n)$ edges, at the cost of at most an additional $O(\log{n})$ factor in the lower bound.

\paragraph{Roadmap.} Section~\ref{diam} contains our lower bound for computing the exact or approximate diameter. In Sections ~\ref{rad},~\ref{ecc} and~\ref{spanners}, we give our lower bounds for computing the exact or approximate radius, for computing eccentricities, and for verifying spanners, respectively. Section~\ref{sec:deg} uses our degree-reduction technique to show a graph construction with constant degree.

\section{Computing the Diameter}\label{diam}

In this section we present lower bounds on the number of rounds needed to compute the diameter \emph{exactly} and \emph{approximately} in sparse networks. First, in Section~\ref{exDiam} we present a higher lower bound on the number of rounds needed for any algorithm to compute the exact diameter of a sparse network, and next, in Section~\ref{appDiam} we show how to modify our sparse construction to achieve a higher lower bound on the number of rounds needed for any algorithm to compute a $(\frac{3}{2}-\varepsilon)$-approximation to the diameter.

\subsection{Exact Diameter}\label{exDiam}

\begin{figure}[h]
	\begin{center}
		\includegraphics[width=14cm,height=8cm]{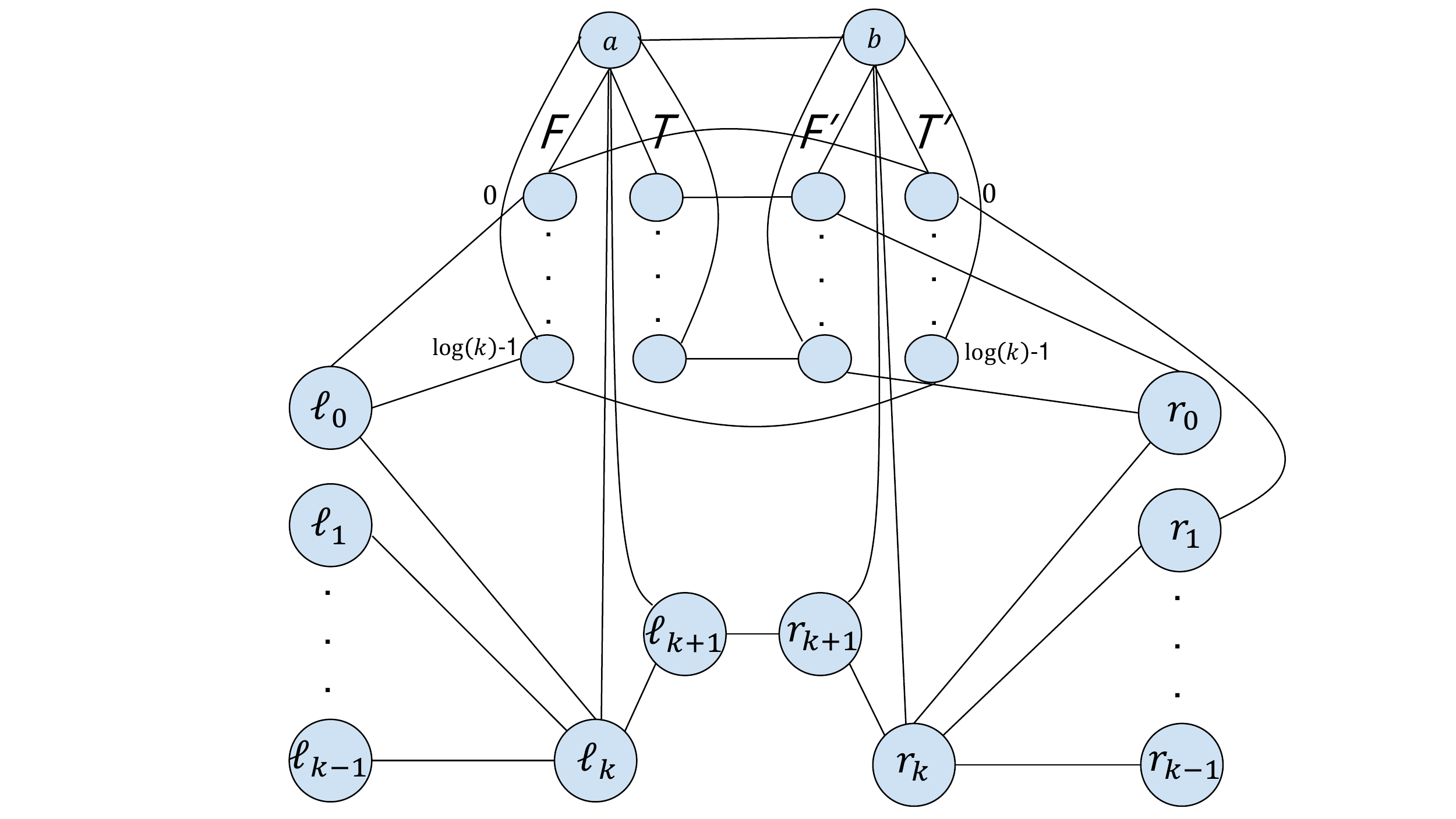}
	\end{center}
	\caption{Graph\ Construction (diameter). Some edges are omitted, for clarity.}
\end{figure}

\begin{theorem-repeat}{thm:sparse}
	\ThmSpa
\end{theorem-repeat}

To prove Theorem~\ref{thm:sparse} we describe a graph construction $G=(V,E)$ and a partition of $G$ into $(G_a,G_b)$, such that one part is simulated by Alice (denoted by $G_a$), and the second is simulated by Bob (denoted by $G_b$). Each player receives an input string defining some additional edges that will affect the diameter of $G$. The proof is organized as follows: in Section~\ref{exDiamCons} we describe the graph construction, and next, in Section~\ref{exDiamRe}, we describe the reduction from the Set-Disjointness problem and deduce Theorem~\ref{thm:sparse}.

\subsubsection{Graph construction}\label{exDiamCons}
Let $i^j$ denote the value of the bit $j$ in the binary representation of $i$. The set of nodes $V$ is defined as follows (see also Figure 1):\footnote{Note that for the sake of simplicity, some of the edges are omitted from Figure 1.} First, it contains two sets of nodes $L= \{\ell_i \mid 0\leq i\leq k-1\}$ and $R= \{r_i \mid 0\leq i\leq k-1\}$, each of size $k$. All the nodes in $L$ are connected to an additional node $\ell_{k}$, which is connected to an additional node $\ell_{k+1}$. Similarly, all the nodes in $R$ are connected to an additional node $r_{k}$, which is connected to an additional node $r_{k+1}$. The nodes $\ell_{k+1}$ and $r_{k+1}$ are also connected by an edge.

Furthermore, we add four sets of nodes, which are our bit-gadget: $F=\{f_j \mid 0\leq j\leq \log(k)-1\},T=\{t_j \mid 0\leq j\leq \log(k)-1\},F'=\{f'_j \mid 0\leq j\leq \log(k)-1\},T'=\{t'_j \mid 0\leq j\leq \log(k)-1\}$, each of size $\log(k)$. We connect the sets $F,T$ with $F',T'$ by adding edges between $f_i$ and $t_i'$, and between $t_i$ and $f_i'$, for each $0\leq i \leq \log(k)-1$.
To define the connections between the sets $L,R$ and the sets $F,T,F',T'$, we add the following edges: For each $\ell_i\in L$, if $i^j=0$, we connect $\ell_i$ to $f_j$, otherwise, we connect $\ell_i$ to $t_j$. Similarly, for each $r_i\in R$, if $i^j=0$ we connect $r_i$ to $f'_j$, otherwise, we connect $r_i$ to $t_j'$.

To complete the construction we add two additional nodes $\{a,b\}$. We connect $a$ to all the nodes in $F\cup T\cup \{\ell_{k},\ell_{k+1}\}$, and similarly, we connect $b$ to all the nodes in $F'\cup T'\cup \{r_{k},r_{k+1}\}$. We also add an edge between the nodes $a$ and $b$.

\begin{claim}\label{OuterNodes}
	
	For every $i,j\in [k-1]$ it holds that $d(\ell_i,r_j)=3$ if $i\neq j$, and $d(\ell_i,r_j)=5$ otherwise.
	
	\begin{proof}
		If $i\neq j$, there must be some bit $h$, such that $i^h\neq j^h$. Assume without loss of generality that $i^h=1$ and $j^h=0$. Then, $\ell_i$ is connected to $t_h$ and $r_j$ is connected to $f'_h$. Since $t_h$ and $f'_h$ are connected by an edge, $d(\ell_i,r_j)=3$.
		For the second part of the claim, note that there are 4 options for any shortest path from $\ell_i$ to $r_i$:
		\begin{enumerate}
			\item Through the nodes $\ell_{k},\ell_{k+1},r_{k+1},r_{k}$.
			\item Through some other node $\ell_j\in L$ such that $j\neq i$ in 2 steps, and then using the shortest path of length 3 between $\ell_j$ and $r_i$.
			\item Through some other node $r_j\in R$ such that $j\neq i$ in 3 steps, and then using the shortest path of length 2 between $r_j$ and $r_i$.
			\item Through at least one of the nodes $a,b$. Note that $d(\ell_i,a)=2$ and $d(\ell_i,b)=3$. Similarly $d(b,r_i)=2$ and $d(a,r_i)=3$.
		\end{enumerate}	
Thus, any possible shortest path between $\ell_i$ and $r_i$ must have length 5.	
	\end{proof}
	
\end{claim}

\begin{claim}
	
	For every $u,v\in V\setminus (L\cup R)$ it holds that $d(u,v) \leq 3$.
	
	\begin{proof}
		
		Let $u,v\in V\setminus (L\cup R)$. By definition, $u$ is connected to one of the nodes in $\{a,b\}$. The same holds for $v$ and since $a$ and $b$ are connected by an edge, $d(u,v) \leq 3$.		
	\end{proof}
	
\end{claim}

\begin{corollary}\label{InnerNodes}
	For every $u,v$ such that $u\in (V_a\setminus L)$ or $v\in (V_b\setminus R)$, it holds that $d(u,v) \leq 4$.
\end{corollary}	

\subsubsection{Reduction from Set-Disjointness}\label{exDiamRe}
To prove Theorem~\ref{thm:sparse}, we show a reduction from the Set-Disjointness problem.
Following the construction defined in the previous section, we define a partition
$(G_a=(V_a,E_a),G_b=(V_b,E_b))$:
\begin{flalign}
\nonumber
&V_a = L\cup F\cup T\cup \{\ell_{k},\ell_{k+1},a\},E_a = \{(u,v)|u,v\in V_a \wedge (u,v)\in E \}\\\nonumber
\nonumber
&V_b = R\cup F'\cup T'\cup \{r_{k},r_{k+1},b\},E_b = \{(u,v)|u,v\in V_b \wedge (u,v)\in E \}\nonumber
\end{flalign}
The graph $G_a$ is simulated by Alice and the graph $G_b$ is simulated by Bob, i.e., in each round, all the messages that nodes in $G_a$ send to nodes in $G_b$ are sent by Alice to Bob. Bob forwards these messages to the corresponding nodes in $G_b$. All the messages from nodes in $G_b$ to nodes in $G_a$ are sent in the same manner. Each player receives an input string $(S_a$ and $S_b)$ of $k$ bits. If the bit $S_a[i] = 0$, Alice adds an edge between the nodes $\ell_i$ and $\ell_{k+1}$. Similarly, if $S_b[i] = 0$, Bob adds an edge between the nodes $r_i$ and $r_{k+1}$.

\begin{observation}\label{nonCrossingDistances}
	For every $u,v\in V_a$, it holds that $d(u,v)\leq 4$. Similarly, $d(u,v)\leq 4$ for every $u,v\in V_b$.
\end{observation}

This is because $d(u_a,\ell_{k+1})\leq 2$ for any $u_a\in V_a$, and $d(u_b,r_{k+1})\leq 2$ for any node $u_b\in V_b$.

\begin{lemma}\label{lemmExDiam}
	
	The diameter of $G$ is at least 5 iff the sets of Alice and Bob are not disjoint.
	
	\begin{proof}
		
		Consider the case in which the sets are disjoint i.e., for every $0\leq i\leq k-1$ either $S_a[i]=0$ or $S_b[i]=0$. We show that for any $u,v \in V$, it holds that $d(u,v)\leq 4$. There are 3 cases:
		\begin{enumerate}
			\item $u\in V_a$ and $v\in V_a$: By Observation~\ref{nonCrossingDistances}, $d(u,v)\leq 4$.
			\item $u\in V_b$ and $v\in V_b$: By Observation~\ref{nonCrossingDistances}, $d(u,v)\leq 4$.
			\item $u\in V_a$ and $v\in V_b$: Consider the case in which $u\in (V_a\setminus L)$ or $v\in (V_b\setminus R)$. By Corollary~\ref{InnerNodes} $d(u,v)\leq 4$. Otherwise, $u\in L$ and $v\in R$, and by Claim~\ref{OuterNodes} $d(\ell_i,r_j)=3$ for every $i\neq j$. Note that in case $i=j$ it holds that $d(\ell_i,r_i) \leq 4$, since either $\ell_i$ is connected to $\ell_{k+1}$ or $r_i$ is connected to $r_{k+1}$, and the shortest path between $\ell_i$ and $r_i$ is the one through the nodes $\ell_{k+1},r_{k+1}$.
		\end{enumerate}	
		In case the two sets are not disjoint, there is some $0\leq i\leq k-1$ such that $S_a[i] = 1$ and $S_b[i] = 1$. Note that in this case $\ell_i$ is not connected directly by an edge to $\ell_{k+1}$ and $r_i$ is not connected directly by an edge to $r_{k+1}$. Therefore, there are only 4 options for any shortest path from $\ell_i$ to $r_i$, which are stated in Claim~\ref{OuterNodes}, from which we get that in this case the diameter of $G$ is at least 5.	
	\end{proof}
\end{lemma}

\paragraph{Proof of Theorem~\ref{thm:sparse}} From Lemma~\ref{lemmExDiam}, we get that any algorithm for computing the exact diameter of the graph $G$ can be used to solve the Set-Disjointness problem. Note that since there are $O(\log(k))$ edges in the cut ($G_a,G_b$), in each round Alice and Bob exchange $O(\log(k)\cdot \log(n))$ bits. Since $k = \Omega(n)$ we deduce that any algorithm for computing the diameter of a network must spend ${\Omega}(\frac{n}{\log^2(n)})$ rounds, and since $|E|=O(n\log(n))$ this bound holds even for sparse networks.

\subsection{$(\frac{3}{2}-\varepsilon)$-approximation to the Diameter}\label{appDiam}
\begin{figure}[h]
	\begin{center}
		\includegraphics[width=14cm,height=8cm]{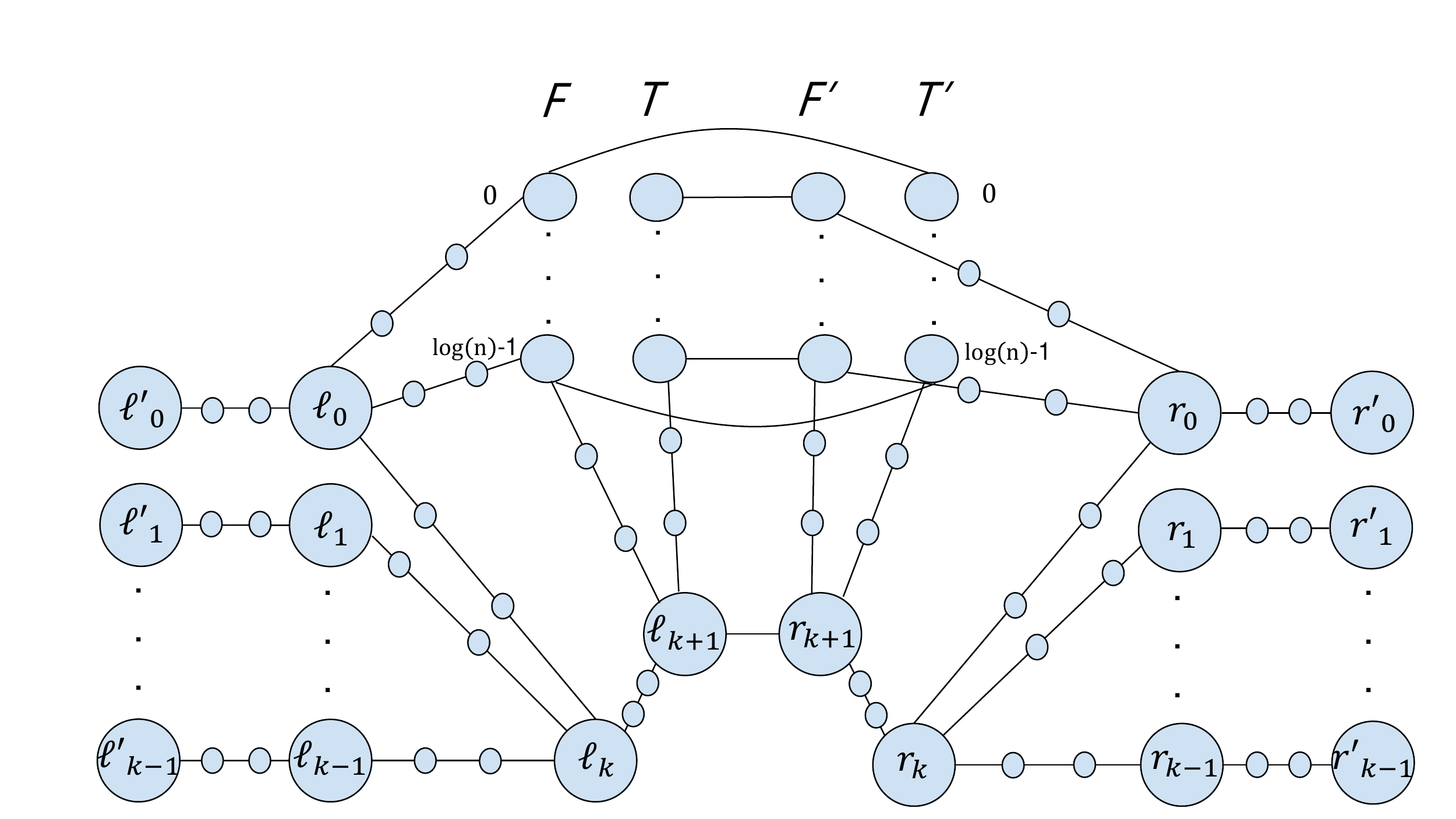}
	\end{center}
	\caption{Graph construction, $P=3$ (diameter approximation). Some edges are omitted.}%, for clarity.}
\end{figure}

 In this Section we show how to modify our sparse construction to achieve a stronger lower bound on the number of rounds needed for any ($\frac{3}{2}-\varepsilon$)-approximation algorithm.

\begin{theorem-repeat}{thm:DA}
	\ThmDA
\end{theorem-repeat}

\subsubsection{Graph Construction}\label{appDiamCons}
The main idea to achieve this lower bound is to stretch our sparse construction by replacing some edges by paths of length $P$, an integer which will be chosen later. Actually, we only apply the following changes to the construction described in Section~\ref{exDiamCons} (see also Figure 2 where $P=3$):
\begin{enumerate}
	\item Remove the nodes $a,b$ and their incident edges.
	\item Replace all the edges incident to the nodes $\ell_{k},r_{k}$ by paths of length $P$.
	\item Replace all the edges $(u,v)$ such that $u\in L$ and $v\in (F\cup T)$ by paths of length $P$. Similarly, Replace all the edges $(u,v)$ such that $u\in R$ and $v\in (F'\cup T')$ by paths of length $P$.
	\item Add two additional sets $L' = \{\ell'_i \mid 0\leq i\leq k-1\}$, $R'=\{r'_i \mid 0\leq i\leq k-1\}$ each of size $k$. Connect each $\ell_i'$ to $\ell_i$, and each $r_i'$ to $r_i$, by a path of length $P$.
\end{enumerate}
Furthermore, to simplify our proof, we connect each $u\in (F\cup T)$ to $\ell_{k+1}$ by a path of length $P$. Similarly, connect each $u\in (F'\cup T')$ to $r_{k+1}$ by a path of length $P$.

\begin{definition}(Y(u,v))
	For each $u,v\in V$ such that $u$ and $v$ are connected by a path of length $P$, denote by $Y(u,v)$ the set of all nodes on the $P$ path between $u$ and $v$ (without $u$ and $v$).
\end{definition}

\begin{claim}\label{InnerNodesApp}
	For every $u,v\in V\setminus (L'\cup R'\bigcup_{i\in [k-1]} Y(\ell'_i,\ell_i)\bigcup_{i\in [k-1]} Y(r'_i,r_i))$ it holds that $d(u,v)$ is at most $4P+1$.
	\begin{proof}
		By the construction, the distance from $u$ to one of the nodes in $\{\ell_{k+1},r_{k+1}\}$ is at most $2P$. The same holds for $v$. Thus, $d(u,v)\leq 4P+1$.
	\end{proof}
\end{claim}

\begin{claim}\label{OuterNodesApp}
	For every $i,j\in [k-1]$ it holds that $d(\ell'_i,r'_j) = 4P+1$ if $i\neq j$, and $d(\ell'_i,r'_j) = 6P+1$ otherwise.
\begin{proof}
	If $i\neq j$, there must be some bit $h$, such that $i^h\neq j^h$. Assume without loss of generality that $i^h=1$ and $j^h=0$. This implies that $\ell_i$ is connected to $t_h$ by a path of length $P$ and $r_j$ is connected to $f'_h$ by a path of length $P$ as well. Since $t_h$ and $f'_h$ are connected by an edge, $d(\ell_i,r_j)=2P+1$ and $d(\ell'_i,r'_j)=4P+1$. For the second part of the claim note that there are 3 options for any shortest path from $\ell'_i$ to $r'_i$:
	
	\begin{enumerate}
		\item Through the node $\ell_{k+1}$ in $3P$ steps, and then using the shortest path of length $3P+1$ between $\ell_{k+1}$ and $r'_i$.
		\item Through some other node $\ell_j\in L$, such that $j\neq i$ in $3P$ steps, and then using the shortest path of length $3P+1$ between $\ell_j$ and $r'_i$.
		\item Through some other node $r_j\in R$, such that $j\neq i$ in $3P+1$ steps, and then using the shortest path of length $3P$ between $r_j$ and $r'_i$.
	\end{enumerate}
	
	Thus, for the second part of the claim, any possible shortest path between $(\ell'_i,r'_i)$ must have length $6P+1$.
\end{proof}
\end{claim}

\subsubsection{Reduction from Set-Disjointness}\label{appDiamRe}
Following the above construction, we define a partition
$(G_a=(V_a,E_a),G_b=(V_b,E_b))$:
\begin{flalign}
\nonumber
V_a =&
\bigcup_{\substack{i\in [k-1] \\ j\in [\log(k)-1] \\ i^j=0}}Y(\ell_i,f_j)\bigcup_{\substack{i\in [k-1] \\ j\in [\log(k)-1] \\ i^j=1}}Y(\ell_i,t_j)\bigcup_{i\in [k-1]} Y(\ell'_i,\ell_i)\\\nonumber
&\bigcup_{i\in [k-1]} Y(\ell_i,\ell_{k})\cup Y(\ell_{k},\ell_{k+1})\cup L'\cup L\cup F\cup T\cup \{\ell_{k},\ell_{k+1}\}\\\nonumber
E_a =& \{(u,v)|u,v\in V_a \wedge (u,v)\in E \}\\\nonumber
V_b =& \bigcup_{\substack{i\in [k-1] \\ j\in [\log(k)-1] \\ i^j=0}}Y(r_i,f'_j)\bigcup_{\substack{i\in [k] \\ j\in [\log(k)-1] \\ i^j=1}}Y(r_i,t'_j)\bigcup_{i\in [k-1]} Y(r'_i,r_i)\\\nonumber
&\bigcup_{i\in [k-1]} Y(r_i,r_{k})\cup Y(r_{k},r_{k+1})\cup R'\cup R\cup F'\cup T'\cup \{r_{k},r_{k+1}\}\\\nonumber
E_b = &\{(u,v)|u,v\in V_b \wedge (u,v)\in E \}\nonumber
\end{flalign}
Each player receives an input string $(S_a$ and $S_b)$ of $k$ bits. If $S_a[i] = 0$, Alice adds an edge between the nodes $\ell_i$ and $\ell_{k+1}$. Similarly, if $S_b[i] = 0$, Bob adds an edge between the nodes $r_i$ and $r_{k+1}$.

\begin{claim}\label{disjoint}
	Let $0\leq i\leq k-1$ be such that $S_a[i]=0$ or $S_b[i]=0$. Then the distance from the node $\ell_i \in L$ to any node $u\in (R\cup \{r_{k+1}\})$ is at most $2P+2$.
\begin{proof}	
	There are 3 cases:
	\begin{enumerate}
		\item $u=r_{k+1}$: Note that $d(\ell_i,\ell_{k+1})\leq 2P$ and $d(\ell_i,r_{k+1})\leq 2P+1$.
		\item $u=r_j\in R$ and $j\neq i$: By Claim~\ref{OuterNodesApp} $d(\ell_i,r_j)=2P+1$.
		\item $u=r_j\in R$ and $j=i$: Note that either $\ell_i$ is connected to $\ell_{k+1}$ directly by an edge or $r_i$ is connected to $r_{k+1}$ directly by an edge. Thus, one of the distances $d(\ell_i,r_{k+1}),d(\ell_{k+1},r_i)$ must be equal to $2$, and both $d(\ell_i,\ell_{k+1}),d(r_{k+1},r_i)$ are at most equal to $2P$, thus,  $d(\ell_i,r_i)\leq 2P+2$.
	\end{enumerate}
\end{proof}
\end{claim}
Note that any node in $V_b$ is connected by a path of length at most $P$ to some node in $R\cup \{r_{k+1}\}$, and any node in $L'$ is connected by a path of length $P$ to some node in $L$. Combining this with Claim~\ref{disjoint} gives the following.

\begin{corollary}\label{disjointCor}
		Let $0\leq i\leq k-1$ be such that $S_a[i]=0$ or $S_b[i]=0$. Then $d(u,v_b)\leq 4P+2$ for any $u\in \{\ell'_i\}\cup Y(\ell'_i,\ell_i)$ and any $v_b\in V_b$. Symmetrically, $d(u,v_a)\leq 4P+2$ for any $u\in \{r'_i\}\cup Y(r'_i,r_i)$ and any $v_a\in V_a$.
\end{corollary}

\begin{lemma}
	
	The Diameter of G is 6P+1 if the two sets of Alice and Bob are not disjoint, and 4P+2 otherwise.
\end{lemma}
\begin{proof}
		Consider the case in which the two sets are not disjoint i.e., there is some $0\leq i\leq k-1$ such that $S_a[i] = 1$ and $S_b[i] = 1$. Note that in this case $\ell_i$ is not connected directly by an edge to $\ell_{k+1}$ and $r_i$ is not connected directly by an edge to $r_{k+1}$. Thus, there are only 3 options for any shortest path from $\ell'_i$ to $r'_i$, both are stated in Claim~\ref{OuterNodesApp}, which implies that in this case $d(\ell'_i,r'_i)=6P+1$. Therefore, the diameter of $G$ is at least $6P+1$.
		Consider the case in which the sets are disjoint i.e., for each $0\leq i\leq k-1$, either $S_a[i]=0$ or $S_b[i]=0$, we need to prove that for every $u,v \in V$, it holds that $d(u,v)\leq 4P+2$. There are 3 cases:
		\begin{enumerate}
			\item $u\in V_a$ and $v\in V_a$: Note that $d(u,\ell_k)\leq 2P$, the same holds for $v$. Thus $d(u,v)\leq 4P$.
			\item $u\in V_b$ and $v\in V_b$: Note that $d(u,r_k)\leq 2P$, the same holds for $v$. Thus $d(u,v)\leq 4P$.
			\item $u\in V_a$ and $v\in V_b$: Consider the case in which $u\in (V_a\setminus (L'\bigcup_{i\in [k-1]} Y(\ell'_i,\ell_i))$ and $v\in (V_b\setminus (R'\bigcup_{i\in [k-1]} Y(r'_i,r_i)))$. By Claim~\ref{InnerNodesApp} $d(u,v)\leq 4P+1$. Otherwise, by Corollary~\ref{disjointCor} $d(u,v)\leq 4P+2$.
		\end{enumerate}	
\end{proof}
\paragraph{Proof of Theorem~\ref{thm:DA}} To complete the proof we need to choose $P$ such that $(\frac{3}{2}-\varepsilon)\cdot (4P+2) < (6P+1)$, this holds for any $P>\frac{1}{2\varepsilon}-\frac{1}{2}$. Note that $k=\Omega(\frac{n}{\log(n)})$ for a constant $\varepsilon$. Thus, we deduce that any algorithm for computing $(\frac{3}{2}-\varepsilon)$-approximation to the diameter requires at least $\Omega(\frac{n}{\log^3(n)})$ rounds. Furthermore, the number of nodes and the number of edges are both equal to $\Theta(k\log(k)\cdot P)$. Thus, this lower bound holds even for graphs with linear number of edges.

	\section{Computing the Radius}\label{rad}
	
	In this section we present lower bounds on the number of rounds needed to compute the radius \emph{exactly} and \emph{approximately} in sparse networks.
	
	\subsection{Exact Radius}

	\begin{theorem}\label{ExactRad}
		The number of rounds needed for any protocol to compute the radius of a sparse network of constant diameter in the $CONGEST$ model is $\Omega(n/\log^2{n})$.
	\end{theorem}

	As in the previous sections, first we describe a graph construction, and then apply a reduction from the Set-Disjointness problem.
	
	\subsubsection{Graph construction}\label{exRadCons}
	
	\begin{figure}[h]
		\begin{center}
			\includegraphics[scale=0.6]{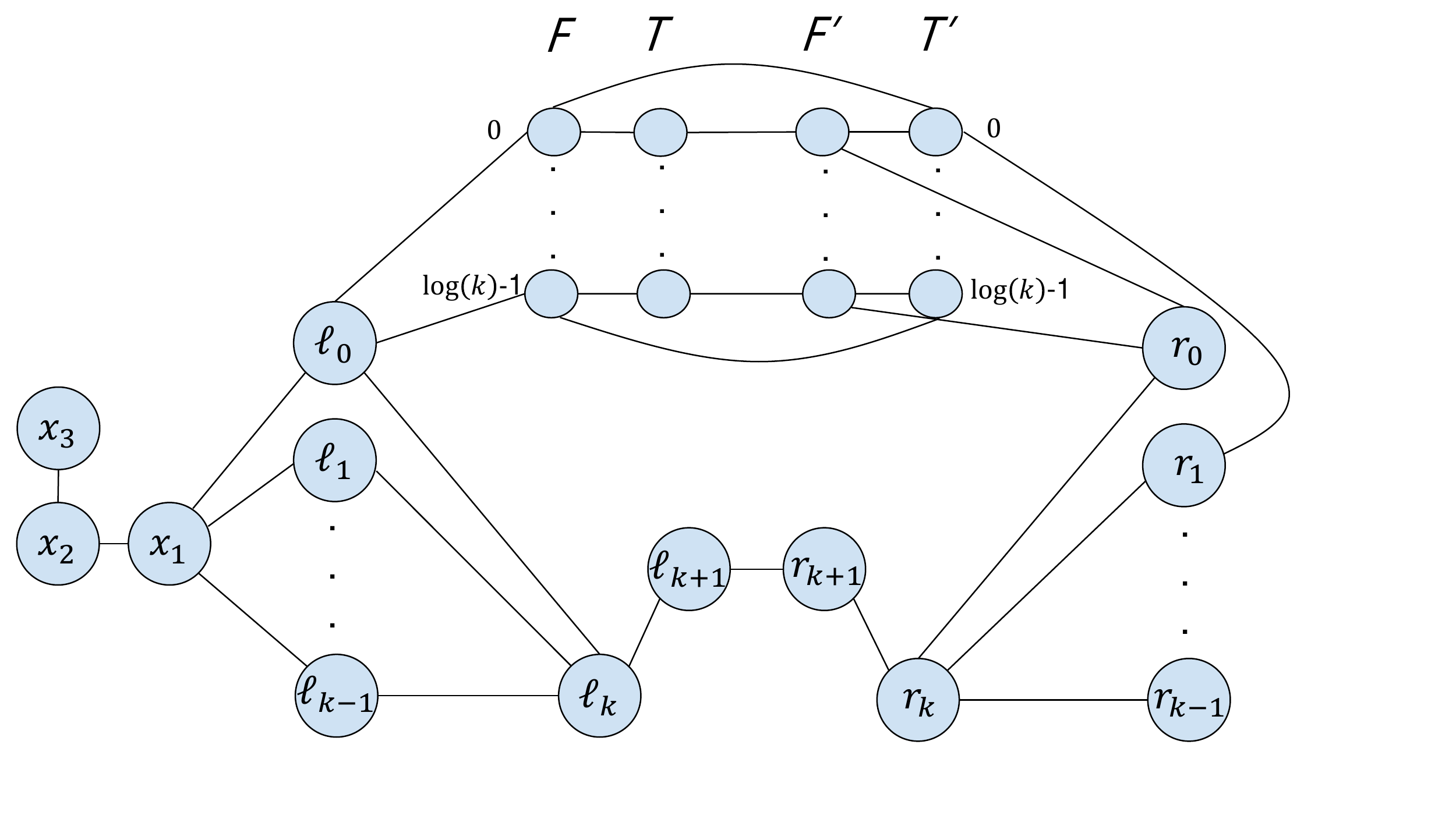}
		\end{center}
		\caption{Graph\ Construction (radius).}
	\end{figure}
	
	The graph construction for the radius is very similar to the one described in Section~\ref{exDiamCons}. We only apply the following changes to the construction described in Section~\ref{exDiamCons} (see also Figure 3):\footnote{Note that for the sake of simplicity, some of the edges are omitted from Figure 3.}
	
	\begin{enumerate}
		\item Remove the nodes $a,b$ and their incident edges.
		\item For each $0\leq i \leq \log(k)-1$ we add an edge between the nodes $f_i\in F$ and $t_i\in T$. Similarly, we add an edge between $f_i'\in F'$ and $t_i'\in T'$.
		\item Add a small gadget which connects each $\ell_i\in L$ to a new node $x_3$ by a path of length 3 $(x_1,x_2,x_3)$.
	\end{enumerate}

	\begin{claim}\label{ex-Ra-O-nodes}
		
		For every $i,j\in [k-1]$ it holds that $d(\ell_i,r_j)=3$ if $i\neq j$, and $d(\ell_i,r_j)=4$ otherwise.
		
		\begin{proof}
			
			If $i\neq j$, there must be some bit $h$, such that $i^h\neq j^h$. Assume without loss of generality that $i^h=1$ and $j^h=0$. Then, $\ell_i$ is connected to $t_h$ and $r_j$ is connected to $f'_h$. Since $t_h$ and $f'_h$ are connected by an edge, $d(\ell_i,r_j)=3$.
			For the second part of the claim, note that all the bits in the binary representation of $i$ are the same as the bits in the binary representation of $j$. Thus, for any edge connecting $\ell_i$ to some node $f_h\in F$, the corresponding edge from $r_j$ is connected to the node $f'_h\in F'$ and not to $t'_h\in T'$. Similarly, for any edge connecting $\ell_i$ to some node $t_h\in T$, the corresponding edge from $r_j$ is connected to the node $t'_h\in T'$ and not to $f'_h\in F'$. Note that for every $0\leq h\leq \log(k)-1$ the nodes $f_h$ and $f'_h$ are not connected directly by an edge and $d(f_h,f'_h)=2$ (the same holds for $t_h$ and $t'_h$). Therefore, $d(\ell_i,r_j)=4$ if $i=j$.
		\end{proof}
		
	\end{claim}
	
	\subsubsection{Reduction from Set-Disjointness}
	
	The reduction is similar to the one used to prove Theorem~\ref{thm:sparse}. First we define a partition $(G_a=(V_a,E_a),G_b=(V_b,E_b))$:
	
	\begin{flalign}
	\nonumber
	&V_a = L\cup F\cup T\cup \{\ell_{k},\ell_{k+1},x_1,x_2,x_3\}\\\nonumber
	&E_a = \{(u,v)|u,v\in V_a \wedge (u,v)\in E \}\\\nonumber
	\nonumber
	&V_b = R\cup F'\cup T'\cup \{r_{k},r_{k+1}\}\\\nonumber
	&E_b = \{(u,v)|u,v\in V_b \wedge (u,v)\in E \}\\\nonumber
	\end{flalign}
	
	Each player receives an input string ($S_a$ and $S_b$) of $k$ bits. If $S_a[i]=1$, Alice adds an edge between the nodes $\ell_i$ and $\ell_{k+1}$. Similarly, if $S_b[i]=1$, Bob adds an edge between the nodes $r_i$ and $r_{k+1}$. Note that unlike the reduction from Section~\ref{exDiamRe}, here, each player adds some edge if the corresponding bit in its input string is 1, rather than 0.
	
	\begin{observation}\label{ex-Ra-min-Ecc-in-L} %Sorry for the (lack of) creativity in the labels :)
		
		For every node $u\in V\setminus L$ it holds that $e(u)\geq 4$.
		\begin{proof}
			There are two cases:
			\begin{enumerate}
				\item $u\in V\setminus (L\cup \{x_1,x_2,x_3\})$: Note that $d(u,x_3)\geq 4$, since the only nodes that are connected to $x_1$ are the nodes in $L$.
				\item $u\in \{x_1,x_2,x_3\})$: Note that $d(u,r_k)\geq 4$
			\end{enumerate}
		\end{proof}
	\end{observation}

	\begin{lemma}\label{ex-Ra-Dist}
		
		The radius of $G$ is 3 if and only if the two sets of Alice and Bob are not disjoint.
		
		\begin{proof}
			
			Consider the case in which the two sets are disjoint. Note that for every $0\leq i\leq k-1$, it holds that $d(\ell_i,r_i) = 4$, since either $\ell_i$ is not connected directly by an edge to $\ell_{k+1}$ or $r_i$ is not connected directly by an edge to $r_{k+1}$. Combining this with Observation~\ref{ex-Ra-min-Ecc-in-L} and Claim~\ref{ex-Ra-O-nodes} we get that for every $u\in V$, $e(u)\geq 4$. Thus, the radius of $G$ is at least 4 as well.
			In case the two sets are not disjoint, i.e., there is some $0\leq i\leq k-1$ such that $S_a[i] = 1$ and $S_b[i] = 1$, we show that $e(\ell_i) = 3$: From Claim~\ref{ex-Ra-O-nodes}, for every $0\leq j\leq k-1$ such that $i\neq j$, it holds that $d(\ell_i,r_j)=3$, and since $S_a[i] = 1$ and $S_b[i] = 1$, it holds that $d(\ell_i,r_i)=3$ as well (the path through the nodes $\ell_{k+1}$ and $r_{k+1}$). It is straightforward to see that for every $u\in V\setminus R$, it holds that $d(\ell_i,u)\leq 3$. Therefore, the radius of $G$ is $\min_{u\in V}e(u)$ which is $e(\ell_i) = 3$.
		\end{proof}
	\end{lemma}
	
	\paragraph{Proof of Theorem~\ref{ExactRad}} From Lemma~\ref{ex-Ra-Dist}, we get that any algorithm for computing the exact radius of the graph $G$ solves the Set-Disjointness problem. Note that as the construction described in Section~\ref{exDiamCons}, there are $O(\log(k))$ edges in the cut ($G_a,G_b$). Thus, in each round Alice and Bob exchange $O(\log(k)\cdot \log(n))$ bits. Since $k = \Omega(n)$ we deduce that any algorithm for computing the exact radius of a network must spend ${\Omega}(\frac{n}{\log^2(n)})$ rounds, and since $|E|=O(n\log(n))$ this lower bound holds even for sparse networks.

	\subsection{$(\frac{3}{2}-\varepsilon)$ approximation to the Radius}
	
	We now show how to extend the construction described in Section~\ref{appDiamCons} to achieve the same asymptotic lower bound for any $(\frac{3}{2}-\varepsilon)$-approximation algorithm to the radius as well.

	\begin{theorem-repeat}{thm:Radius}
		\ThmRadius
	\end{theorem-repeat}

	\subsubsection{Graph construction}\label{appRadCons}
	
	\begin{figure}[h]
		\begin{center}
			\includegraphics[scale=0.6]{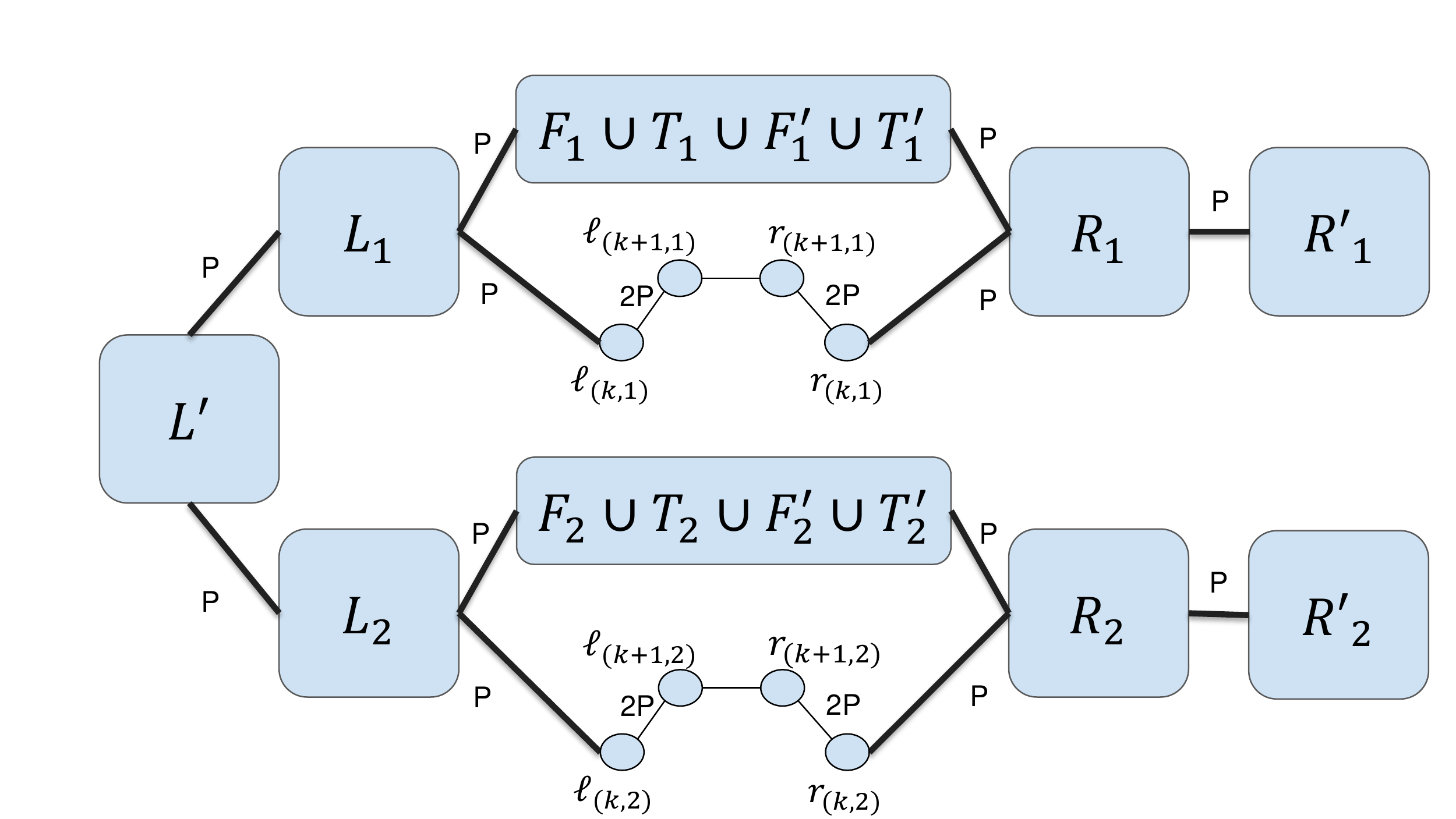}
		\end{center}
		\caption{Graph\ Construction (radius approximation). Each edge between two sets of nodes represents all the edges between the corresponding sets.}
	\end{figure}
	
	Let $G=(V,E)$ be the graph construction described in section~\ref{appDiamCons}. We describe the construction for this section in two steps. In the first step, we apply the following changes to $G$:
	\begin{enumerate}
		\item Replace the path of length $P$ connecting $\ell_{k}$ to $\ell_{k+1}$ by a path of length $2P$. Similarly, replace the path of length $P$ connecting $r_{k}$ to $r_{k+1}$ by a path of length $2P$.
		\item Remove all the paths of length $P$ connecting $\ell_{k+1}$ with some node $u \in (F\cup T)$. Similarly, remove all the paths of length $P$ connecting $r_{k+1}$ with some node $u' \in (F'\cup T')$.
	\end{enumerate}

	Let $V'$ denote the set of vertices $V\setminus (L'\bigcup_{i\in [k-1]} Y(\ell'_i,\ell_i))$, and let $G'$ be the graph induced by the vertices in $V'$. Now we can describe the second step, in which we extend the construction in the following manner: we have two instances of $G'$, denoted by $G'_1=(V'_1,E'_1)$ and $G'_2=(V'_2,E'_2)$. Each node in $G'_1\cup G'_2$ has an additional label, indicating whether it belongs to $G'_1$ or $G'_2$. For example, the node $\ell_i$ in $G'_1$ is denoted by $\ell_{(i,1)}$, and the node $r_i$ in $G'_2$ is denoted by $r_{(i,2)}$. To complete the construction, for each $0\leq i\leq k-1$ we add two paths, each of length $P$, connecting $\ell'_i$ to the nodes $\ell_{(i,1)},\ell_{(i,2)}$.
	
	\subsubsection{Reduction from Set-disjointness}
	The reduction is similar to the one used to prove Theorem~\ref{thm:DA}. First we define a partition $(G_a=(V_a,E_a),G_b=(V_b,E_b))$, let $V'_a$ be the set of vertices defined by:
	
	\begin{flalign}
	\nonumber
	&V'_a = L'\cup L_1\cup L_2\cup F_1\cup F_2\cup T_1\cup T_2\cup \{\ell_{(k,1)},\ell_{(k+1,1)},\ell_{(k,2)},\ell_{(k+1,2)}\}\\\nonumber
	\end{flalign}
	
	Similarly, let $V'_b$ be the set of vertices defined by:
	
	\begin{flalign}
	\nonumber
	&V'_b = R_1\cup R_2\cup R'_1\cup R'_2\cup F'_1\cup F'_2\cup T'_1\cup T'_2\cup \{r_{(k,1)},r_{(k+1,1)},r_{(k,2)},r_{(k+1,2)}\}\\\nonumber
	\end{flalign}
	
	Now we define the sets $V_a,E_a,V_b,E_b$:
	
	\begin{flalign}
	\nonumber
	&V_a = V'_a\bigcup_{\substack{u,v\in V'_a}} Y(u,v) \\\nonumber
	&E_a = \{(u,v)|u,v\in V_a \wedge (u,v)\in E \}\\\nonumber
	\nonumber
	&V_b = V'_b\bigcup_{\substack{u,v\in V'_b}} Y(u,v) \\\nonumber
	&E_b = \{(u,v)|u,v\in V_b \wedge (u,v)\in E \}\\\nonumber
	\end{flalign}
	
	Each player receives an input string ($S_a$ and $S_b$) of $k$ bits. If $S_a[i]=1$, Alice adds two paths, each of length $P$, the first connects the nodes $\ell_{(i,1)}$ and $\ell_{(k+1,1)}$, and the second connects the nodes $\ell_{(i,2)}$ and $\ell_{(k+1,2)}$. Similarly, if $S_b[i]=1$, Bob adds two paths of length $P$, the first connects the nodes $r_{(i,1)}$, and $r_{(k+1,1)}$ and the second connects the nodes $r_{(i,2)}$ and $r_{(k+1,2)}$.
	
	\begin{observation}\label{ObL'}
		The node with the minimum eccentricity is in L'.
	\end{observation}
	
	This is because for any node $u\in V'_1\bigcup_{i\in[k-1]}Y(\ell'_i,\ell_{(i,1)})$ it holds that the node with $\max_{v\in V} d(u,v)$ must be in $V'_2$. Similarly, for any node $u\in V'_2\bigcup_{i\in[k-1]}Y(\ell'_i,\ell_{(i,2)})$ it holds that the node with $\max_{v\in V} d(u,v)$ must be in $V'_1$. Thus, any node $u\in V\setminus L'$ must visit some node in $L'$ in order to reach the node with $\max_{v\in V} d(u,v)$.
	
	\begin{observation}\label{Ob2}
		For any $0\leq i\leq k-1$ and any $u\in V_a$, it holds that $d(\ell'_i,u) \leq 4P$.
	\end{observation}
	
	This is because the distance from any $u\in V_a$ to one of the nodes $\{\ell_{(k,1)},\ell_{(k,2)}\}$ is at most $2P$, and $d(\ell'_i,\ell_{(k,1)})=d(\ell'_i,\ell_{(k,2)})=2P$ for any $0\leq i\leq k-1$.

	\begin{claim}\label{distAppRad}
		Let $0\leq i\leq k-1$ be such that $S_a[i]=0$ or $S_b[i]=0$. It holds that $d(\ell'_i,r'_{(i,1)}) \geq 6P+1$. Similarly, $d(\ell'_i,r'_{(i,2)}) \geq 6P+1$ as well.
		\begin{proof}
			Note that either $\ell_{(i,1)}$ is not connected to $\ell_{(k+1,1)}$ by a path of length $P$ or $r_{(i,1)}$ is not connected to $r_{(k+1,1)}$ by a path of length $P$. Thus, there are three options for any shortest path between $\ell'_i$ and $r'_{(i,1)}$:
			\begin{enumerate}
				\item Through the nodes ${\ell_{(k+1,1)},r_{(k+1,1)}}$: In this case, the shortest path must visit at least one of the nodes ${\ell_{(k,1)},r_{(k,1)}}$. Assume without loss of generality that the shortest path visits only $\ell_{(k,1)}$. Thus, the distance between $\ell'_i$ and $r'_{(i,1)}$ can be written as the sum of two distances: $d(\ell'_i,r'_{(i,1)}) = d(\ell'_i,\ell_{(k,1)}) + d(\ell_{(k,1)},r'_{(i,1)})$. Note that $d(\ell_{(k,1)},r'_{(i,1)})\geq 4P+1$(and equals $4P+1$ if $S_b[i]=1$), and $d(\ell'_i,\ell_{(k,1)})=2P$. Thus, we have $d(\ell'_i,r'_{(i,1)}) \geq 2P + 4P+1 = 6P+1$.
				\item Through some other node $\ell_{(j,1)}\in L_1$ such that $j\neq i$ in $3P$ steps, and then using the shortest path of length $3P+1$ between $\ell_{(j,1)}$ and $r'_{(i,1)}$.
				\item Through some other node $r_{(j,1)}\in L_1$ such that $j\neq i$ in $3P+1$ steps, and then using the shortest path of length $3P$ between $r_{(j,1)}$ and $r'_{(i,1)}$.
			\end{enumerate}
			Thus, $d(\ell'_i,r'_{(i,1)}) = 6P+1$, and similarly, $d(\ell'_i,r'{(i,2)}) = 6P+1$ as well.
		\end{proof}
	\end{claim}
	
	Combining Claim~\ref{distAppRad} with Observation~\ref{ObL'}, we conclude the following.
	
	\begin{corollary}\label{coldist}
		The radius of $G \geq 6P+1$ if the two sets of Alice and Bob are disjoint.
	\end{corollary}

	\begin{claim}\label{8} 
		Let $0\leq i\leq k-1$ be such that $S_a[i]=1$ and $S_b[i]=1$. Then the distance from $\ell'_i$ to any node in $V_b$ is at most $4P+1$.
		\begin{proof}
			Note that $\ell_{(i,1)}$ is connected to $\ell_{(k+1,1)}$ by a path of length $P$ and $r_{(i,1)}$ is connected to $r_{(k+1,1)}$ by a path of length $P$. Similarly, $\ell_{(i,2)}$ is connected to $\ell_{(k+1,2)}$ by a path of length $P$ and $r_{(i,2)}$ is connected to $r_{(k+1,2)}$ by a path of length $P$. Thus, the distance from $\ell'_i$ to each of the nodes $r_{(k+1,1)},r_{(k+1,2)}$ is $2P+1$, and $d(\ell'_i,v)\leq 4P+1$ for any $v\in (Y(r_{(k,1)},r_{(k+1,1)})\cup Y(r_{(k,2)},r_{(k+1,2)}))$. It remains to show that $d(\ell'_i,v)\leq 4P+1$ for any $v\in V_b\setminus (Y(r_{(k,1)},r_{(k+1,1)})\cup Y(r_{(k,2)},r_{(k+1,2)}))$. Note that $\ell'_i$ can reach any node in $R_1\cup R_2$ in $3P+1$ steps. Furthermore, by the definition of the construction, for any $v\in V_b\setminus (Y(r_{(k,1)},r_{(k+1,1)})\cup Y(r_{(k,2)},r_{(k+1,2)}))$ there is some node $u\in (R_1\cup R_2)$, such that $d(v,u)=P$. Thus, for any $v\in V_b\setminus (Y(r_{(k,1)},r_{(k+1,1)})\cup Y(r_{(k,2)},r_{(k+1,2)}))$ it holds that $d(\ell'_i,v) \leq 4P+1$.
		\end{proof}
	\end{claim}
	
	Combining Observation~\ref{ObL'}, Observation~\ref{Ob2}, Corollary~\ref{coldist} and Claim~\ref{8}, we conclude the following.
	
	\begin{lemma}
		The Radius of G is 4P+1 if the two sets of Alice and Bob are not disjoint, and at least 6P+1 otherwise.
	\end{lemma}

	\paragraph{Proof of Theorem~\ref{thm:Radius}} To deduce Theorem~\ref{thm:Radius} we need to choose $P$ such that $(\frac{3}{2}-\varepsilon)\cdot (4P+1) < (6P+1)$, this holds for any $P>\frac{1}{8\varepsilon}-\frac{1}{4}$. Note that $k=\Omega(\frac{n}{\log(n)})$ for a constant $\varepsilon$. Thus, we deduce that any algorithm for computing a $(\frac{3}{2}-\varepsilon)$-approximation to the radius requires at least $\Omega(\frac{n}{\log^3(n)})$ rounds. Furthermore, the number of nodes and the number of edges are both equal to $\Theta(k\log(k)\cdot P)$. Thus, this lower bound holds even for graphs with linear number of edges.

	\subsection{Shaving an Extra Logarithmic Factor from the
		Denominator}
	\begin{figure}[h]
		\begin{center}
			\includegraphics[scale=0.7]{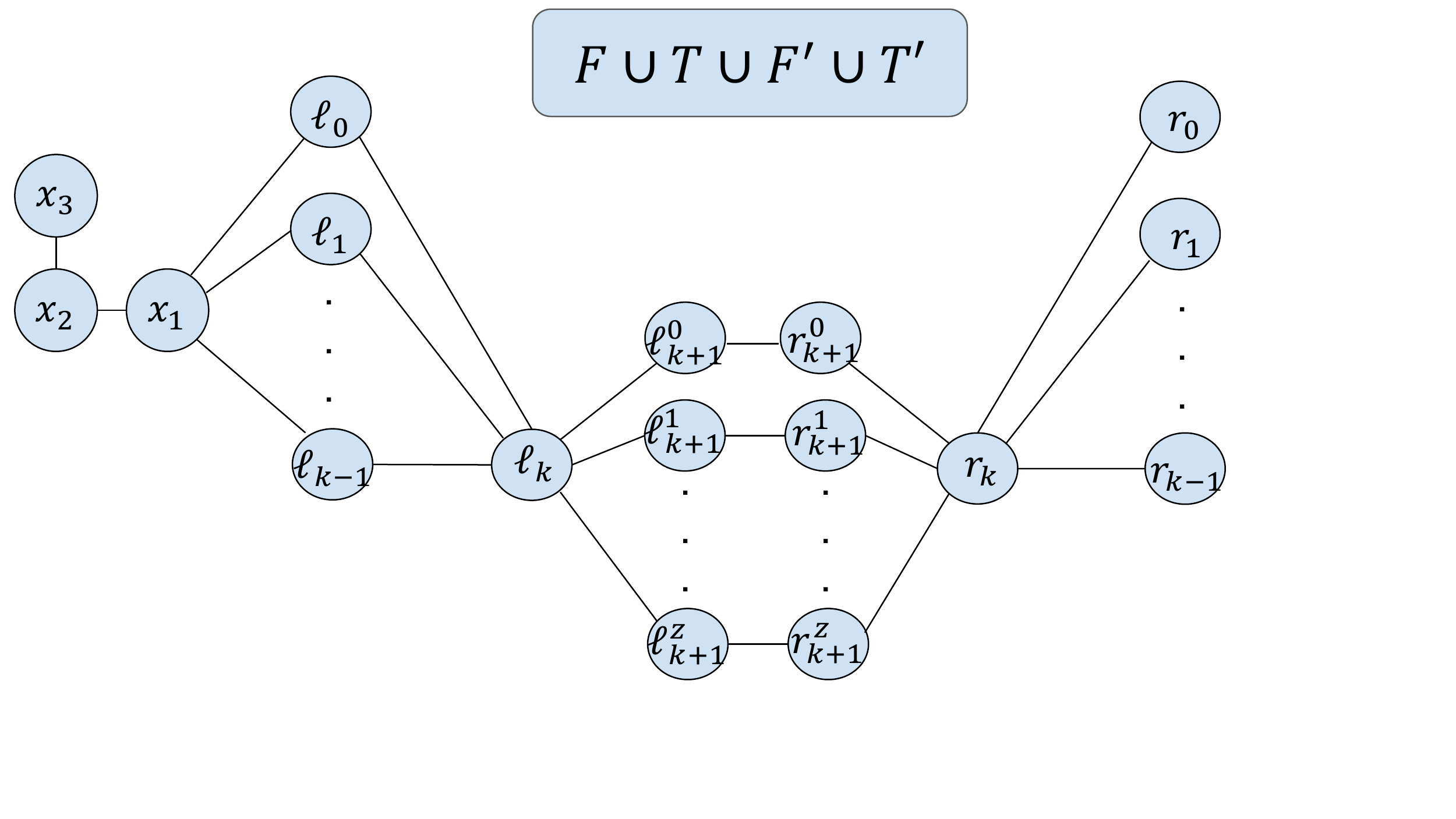}
		\end{center}
		\caption{Graph\ Construction (radius), $z=\log(k)-1$. Some edges are omitted, for clarity.}
	\end{figure}
	
	In this section we show how to modify the construction described in Section~\ref{exRadCons} to achieve higher lower bounds for computing the radius \emph{exactly} and \emph{approximately}. The general idea is to expand the input strings of Alice and Bob while preserving the (asymptotic) size of the cut. We apply the following changes to the construction described in Section~\ref{exRadCons} (see also Figure 5):
	\begin{enumerate}
		\item Replace the node $\ell_{k+1}$ by a set of nodes $\{\ell^j_{k+1} \mid 0\leq j \leq \log(k)-1\}$ of size $\log(k)$. Similarly, replace the node $r_{k+1}$ by a set of nodes $\{r^j_{k+1} \mid 0\leq j \leq \log(k)-1\}$ of size $\log(k)$.
		\item Connect each $\ell^j_{k+1}$ with $r^j_{k+1}$ by an edge.
	\end{enumerate}
	
	For the graph partition, we add to $V_a$ the nodes $\{\ell^j_{k+1} \mid 0\leq j \leq \log(k)-1\}$ and we add to $V_b$ the nodes $\{r^j_{k+1} \mid 0\leq j \leq \log(k)-1\}$. The input
	to the set-disjointness problem now differs from the previous
	construction, as follows. Each of the players Alice and Bob receives an input string $(S_a,S_b)$ of size $k\log(k)$, each bit is represented by two indices $i,j$ such that $i\in \{0,..,k-1\}$ and $j\in \{0,...,\log(k)-1\}$. If the bit $S_a(i,j)=1$ Alice adds an edge between the nodes $\ell_i$ and $\ell^j_{k+1}$. Similarly, If the bit $S_b(i,j)=1$ Bob adds an edge between the nodes $r_i$ and $r^j_{k+1}$. Note that if the two sets of Alice and Bob are disjoint, then for each $0\leq i\leq k-1$ it holds that $d(\ell_i,r_i)=4$. Otherwise, there is some $0\leq i\leq k-1$ such that $d(\ell_i,r_i)=3$. Thus, one can prove by case analysis, that the radius of $G$ is $3$ if and only if the two sets of Alice and Bob are not disjoint. Note that the size of the cut $(G_a,G_b)$ remains $\Theta(\log(k))$, $k$ remains $\Theta(n)$, and the size of the strings is $\Theta(n\log(n))$. Thus, we achieve the following:
	
	\begin{theorem}\label{radShaved}
		The number of rounds needed for any protocol to compute the radius of a sparse network of constant diameter in the $CONGEST$ model is $\Omega(n/\log{n})$.
	\end{theorem}
	
	Similarly, we can apply the same idea to the construction described in Section~\ref{appRadCons} and achieve the following theorem:
	\begin{theorem}\label{radShavedApp}
		The number of rounds needed for any protocol to compute a ($3/2-\varepsilon$)-approximation to the radius of a sparse network is $\Omega(n/\log^2n$).
	\end{theorem}

	\section{Computing a $(\frac{5}{3}-\varepsilon)$-approximation to the Eccentricity}\label{ecc}
	
	In this section we show that any algorithm for computing a $(\frac{5}{3}-\varepsilon)$-approximation to all the eccentricities must spend $\Omega(\frac{n}{\log^3(n)})$ rounds. Similar to the previous sections, we define a graph construction, and then apply a reduction from the set-disjointness problem.
	
	\subsection{Graph Construction}
	
	\begin{figure}[h]
		\begin{center}
			\includegraphics[scale=0.6]{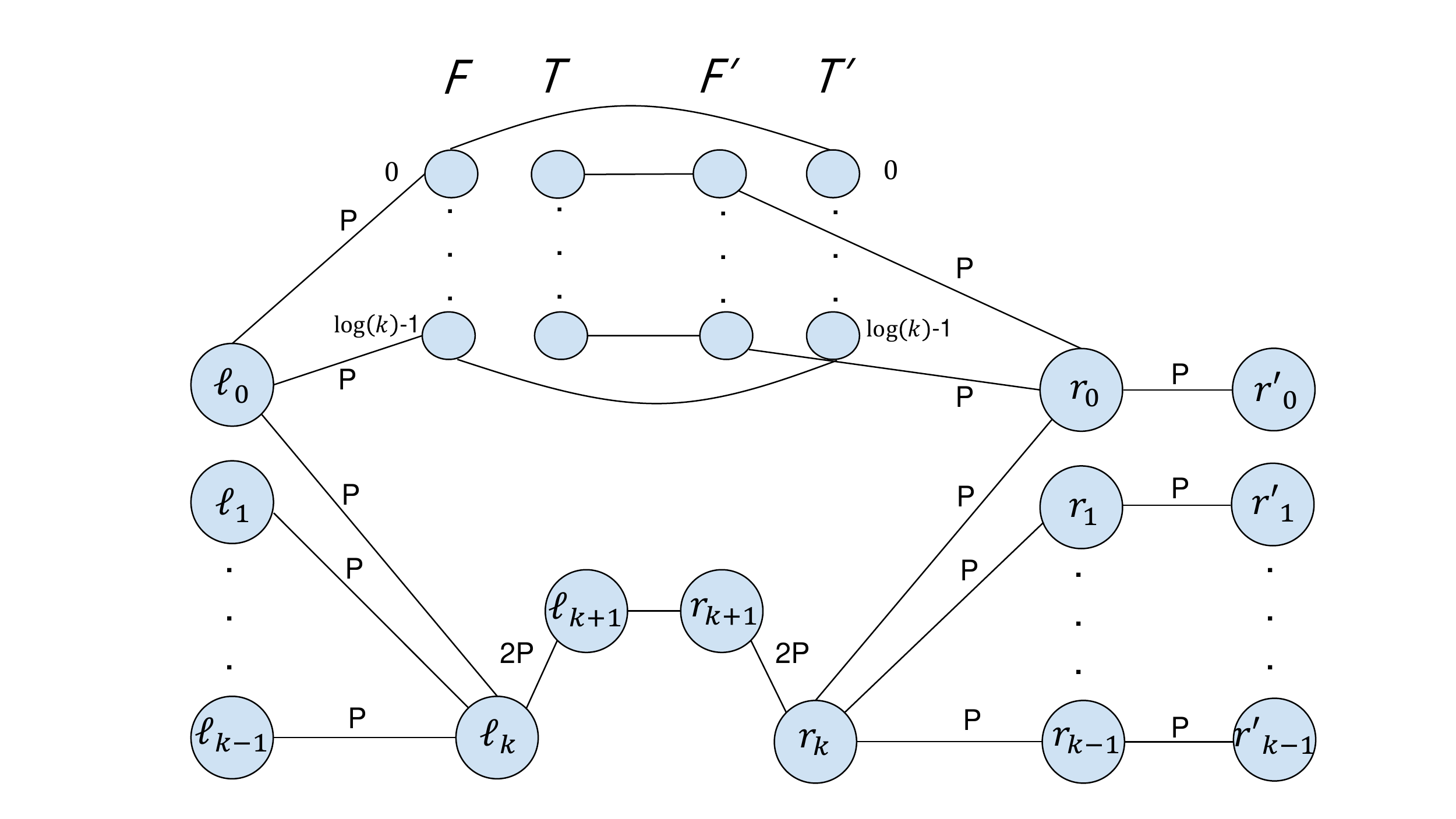}
		\end{center}
		\caption{Graph\ construction ($(\frac{5}{3}-\varepsilon)$-approximation to Eccentricity).}
	\end{figure}
	
	We simply apply the following changes on the construction described in Section~\ref{appDiamCons} (see also Figure 6):
	\begin{enumerate}
		\item Remove all the nodes in $L'\bigcup_{i\in [k-1]} Y(\ell'_i,\ell_i)$ and their incident edges.
		\item Remove all the paths of length $P$ connecting $\ell_{k+1}$ with some node $u \in (F\cup T)$. Similarly, remove all the paths of length $P$ connecting $r_{k+1}$ with some node $u' \in (F'\cup T')$.
		\item Replace the path of length $P$ connecting $\ell_{k}$ to $\ell_{k+1}$ by a path of length $2P$. Similarly, replace the path of length $P$ connecting $r_{k}$ to $r_{k+1}$ by a path of length $2P$.
	\end{enumerate}
	
	\subsection{Reduction from Set-disjointness}
	
	Each player receives an input string ($S_a$ and $S_b$) of $k$ bits. If $S_a[i]=1$, Alice adds a path of length $P$ connecting the nodes $\ell_i$ and $\ell_{k+1}$. Similarly, if $S_b[i]=1$, Bob adds a path of length $P$ connecting the nodes $r_i$ and $r_{k+1}$.
	
	\begin{lemma}
		There is some node $\ell_i\in L$ with eccentricity $3P+1$ if and only if the two sets of Alice and Bob are not disjoint, otherwise, the eccentricity of each node in $L$ is $5P+1$.
		\begin{proof}
			Consider the case that the sets are disjoint i,e., for every $0\leq i\leq k-1$ either $\ell_i$ is not connected to $\ell_{k+1}$ by a path of length $P$ or $r_i$ is not connected to $r_{k+1}$ by a path of length $P$. Thus, there are only 3 options for any shortest path from $\ell_i$ to $r'_i$:
			\begin{enumerate}
				\item ($S_a[i]=0$ and $S_b[i]=0$) Through some other node $\ell_j$ such that $i\neq j$ in $2P$ steps, and then using the shortest path of length $3P+1$ from $\ell_j$ to $r'_i$. Or Through some other node $r_j$ such that $i\neq j$ in $2P+1$ steps, and then using the shortest path of length $3P$ from $r_j$ to $r'_i$.
				\item ($S_a[i]=1$ and $S_b[i]=0$) Through the node $\ell_{k+1}$ in $P$ steps, and then using the shortest path of length $4P+1$ from $\ell_{k+1}$ to $r'_i$.
				\item ($S_a[i]=0$ and $S_b[i]=1$) Through the node $r_{k+1}$ in $3P+1$ steps, and then using the shortest path of length $2P$ from $r_{k+1}$ to $r'_i$.
			\end{enumerate}
			
			In case the two sets are not disjoint, there is some $0\leq i\leq k-1$ such that $\ell_i$ is connected to $\ell_{k+1}$ by a path of length $P$ and $r_i$ is connected to $r_{k+1}$ by a path of length $P$. Thus, the distance from $\ell_i$ to any of the nodes in $R$ is $2P+1$, and to any of the nodes in $R'$ is $3P+1$. It is straightforward to see, by the definition of the construction, that the distance from $\ell_i$ to any node in $V\setminus R'$ is at most $3P+1$ as well.
			
		\end{proof}
	\end{lemma}
	
	\begin{theorem-repeat}{thm:ecc}
		\ThmEcc
	\end{theorem-repeat}
	
	\paragraph{Proof of Theorem~\ref{thm:ecc}} To complete the proof we need to choose $P$ such that $(\frac{5}{3}-\varepsilon)\cdot (3P+1) < (5P+1)$, this holds for any $P>\frac{2}{9\varepsilon}-\frac{1}{3}$. Note that $k=\Omega(\frac{n}{\log(n)})$ for a constant $\varepsilon$. Thus, we deduce that any algorithm for computing a $(\frac{5}{3}-\varepsilon)$-approximation to all the eccentricities requires at least $\Omega(\frac{n}{\log^3(n)})$ rounds.
	
	\section{Networks with $\Delta=5$}
	\label{sec:deg}	
	In this section we show how to modify our construction for radius computation to achieve almost the same lower bound for networks with $\Delta=5$, where $\Delta$ is the maximum degree of the network. Similar modification can be applied to the construction of the diameter as well. Consider the construction described in section ~\ref{ExactRad}. For each node $v\in V$ such that $v$ has a large degree, we replace a subset of the edges incident to $v$ by a binary tree (see also Figure 7). For example, consider the node $\ell_0$. Instead of connecting $\ell_0$ to each of its neighbors in $F\cup T$ directly by edges, we add a binary tree of size $O(\log(k))$, whose root is $\ell_0$ and whose leaves are the corresponding nodes in $F\cup T$. Note that there are $O(k)$ nodes with degree $O(\log(k))$ and $O(\log(k))$ nodes with degree $O(k)$ in the original graph.
	
	\begin{figure}[h]
		\begin{center}
			\includegraphics[scale=0.6]{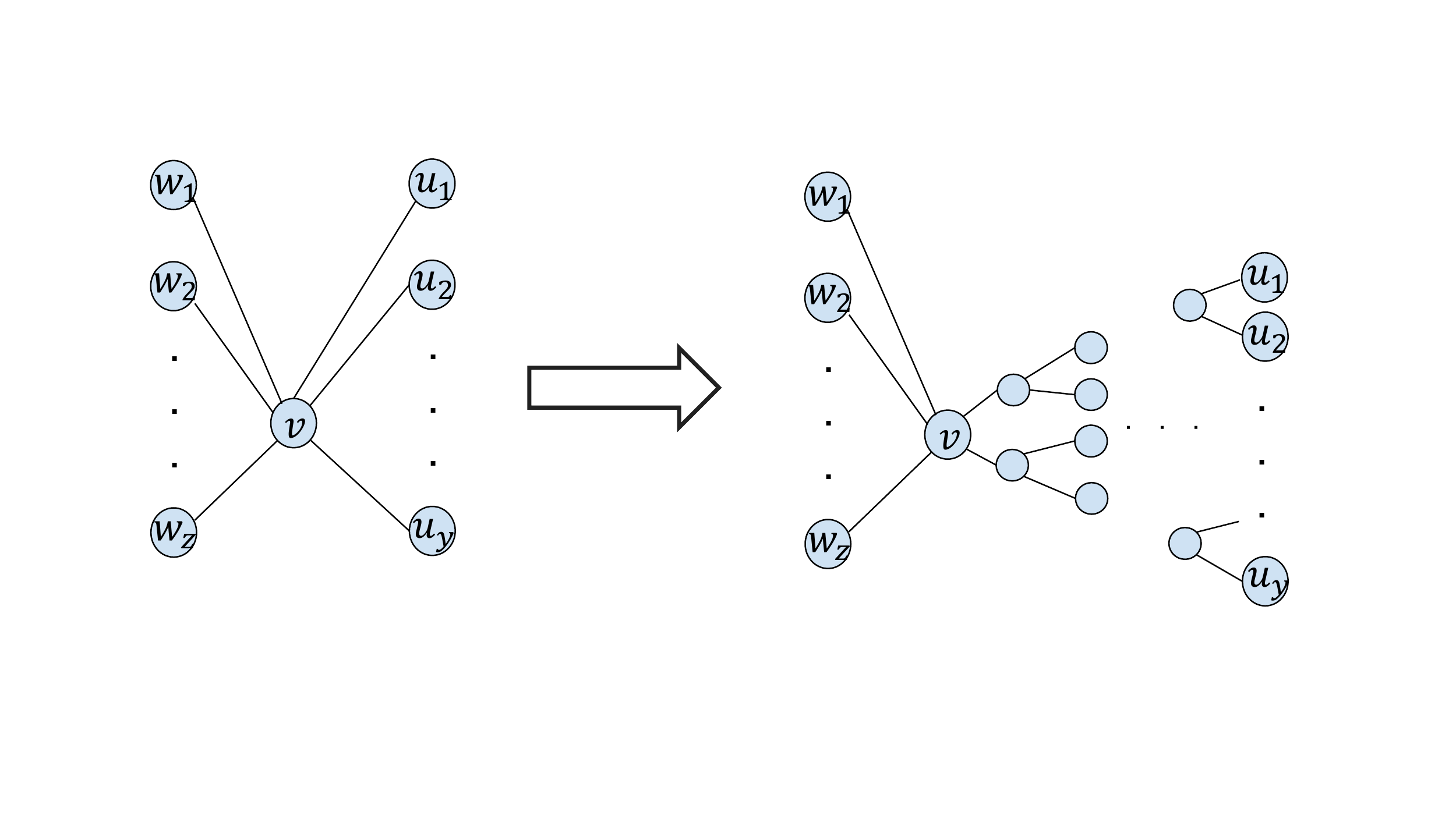}
		\end{center}
		\caption{Replacing $y$ edges connected to some node $v$ of degree $y+z$ by a binary tree of size $O(y)$. Thus, the degree of $v$ becomes $z+2$.}
	\end{figure}
	
	\begin{theorem-repeat}{ExactRadsmallDeg}
		\ThmDeg
	\end{theorem-repeat}
	
	Now we describe the changes we apply to the construction described in Section 3.1.1 formally by 10 steps:
	\begin{enumerate}
		\item Remove the node $x_3$ and its incident edge.
		\item Replace the edges between $x_1$ and $L$ by a binary tree of size $O(k)$ such that $x_1$ is the root and the leaves are the nodes in $L$. Note that the height of this tree is exactly $\log(k)$.
		\item Replace the edge $(x_1,x_2)$ by a path of length $\log(k)+2\log\log(k)-1$. Denote the set of all the nodes on this path, including $x_1$ and $x_2$, by $P(x_1,x_2)$.
		\item Remove all the edges connecting some node in $L\cup R$ with some node in the \emph{bit-gadget}.
		\item Connect each $\ell_i$ in $L$ to all the nodes that represent its binary value by a binary tree of size $O(\log(k))$, such that $\ell_i$ is the root and the leaves are the corresponding nodes in $F\cup T$. Similarly, connect each $r_i$ in $R$ to all the nodes that represent its binary value by a binary tree of size $O(\log(k))$, such that $r_i$ is the root and the leaves are the corresponding nodes in $F'\cup T'$. Note that the height of each such tree is exactly $\log\log(k)$.
		\item After the previous step, each of the nodes in the \emph{bit-gadget} is a leaf in $k$ trees, thus, we replace each $u$ in the \emph{bit-gadget} by a binary tree of size $O(k)$ such that the root is $u$ and the leaves are its $k$ parents in the $k$ binary trees. Note that the height of each such tree is exactly $\log(k)$. Thus, after this step, the distance between $\ell_0$ and $f_0$, for example, is $\log\log(k)+\log(k)-1$ and not $\log\log(k)+\log(k)$. This is because the parents of the leaves in the tree rooted at $\ell_0$ are leaves in the trees rooted at the corresponding nodes in $F\cup T$.
		\item Replace the edges in $\{(\ell_i,\ell_k) \mid 0\leq i\leq k-1\}\cup\{(r_i,r_k) \mid 0\leq i\leq k-1\}$ by paths of length $\log\log(k)-1$. Denote by $P(\ell_k)$ the set of all the nodes on the paths connecting $\ell_k$ to the nodes in $L$, including the nodes in $L$ and the node $\ell_k$. Similarly, denote by $P(r_k)$ the set of all the nodes on the paths connecting $r_k$ to the nodes in $R$, including the nodes in $R$ and the node $r_k$.
		\item Replace the set of the first edges on every path that connects $\ell_k$ to the nodes in $L$ (i.e., the $k$ edges that are incident to $\ell_k$) by a binary tree of size $O(k)$ and height exactly $\log(k)$. Similarly, replace set of the first edges on every path that connects $r_k$ to the nodes in $R$ (i.e., the $k$ edges that are incident to $r_k$) by a binary tree of size $O(k)$ and height exactly $\log(k)$. Note that the edges $(\ell_k,\ell_{k+1})$ and $(r_k,r_{k+1})$ remain the same. Thus, after this step, the distance between $\ell_0$ and $\ell_{k}$, for example, is $\log\log(k)+\log(k)-1$, and thus, the distance between $\ell_0$ and $\ell_{k+1}$ is $\log\log(k)+\log(k)$.
		\item Connect the node $\ell_{k+1}$ to each of the nodes in $L$ by a new path of length $\log\log(k)-1$, for each $\ell_i$. Denote by $q(\ell_i)$ the neighbor of $\ell_i$ in the corresponding path of size $\log\log(k)-1$ connecting $\ell_i$ to $\ell_{k+1}$. Similarly, connect the node $r_{k+1}$ to each of the nodes in $R$ by a new path of length $\log\log(k)-1$, and for each $r_i$, denote by $q(r_i)$ the neighbor of $r_i$ in the corresponding path of size $\log\log(k)-1$ connecting $r_i$ to $r_{k+1}$. Denote by $P(\ell_{k+1})$ the set of all the nodes on the paths connecting $\ell_{k+1}$ to the nodes in $L$, including the nodes in $L$ and the node $\ell_{k+1}$. Similarly, Denote by $P(r_{k+1})$ the set of all the nodes on the paths connecting $r_{k+1}$ to the nodes in $R$, including the nodes in $R$ and the node $r_{k+1}$.
		\item Replace the set of the first edges on every path that connects $\ell_{k+1}$ to the nodes in $L$(i.e., the $k$ edges that are incident to $\ell_{k+1}$) by a binary tree of size $O(k)$ and height exactly $\log(k)$. Similarly, replace the first edges on every path that connects $r_{k+1}$ to the nodes in $R$(i.e., the $k$ edges that are incident to $r_{k+1}$) by a binary tree of size $O(k)$ and height exactly $\log(k)$.
	\end{enumerate}
	
	We note that we keep the edges between the nodes $f_i\in F$ and $t_i\in T$ and the edges between the nodes $f_i'\in F'$ and $t_i'\in T'$ for each $0\leq i \leq \log(k)-1$ as in Section ~\ref{ExactRad}.
	For each $v\in V$ such that $v$ is a root of some binary tree in $G$, denote by $T(v)$ the set of all nodes in the binary tree rooted at $v$ not including $v$ itself.
	
	\begin{claim}
		The maximum degree in the graph is 5.
		\begin{proof}
			We show that for each $v\in V$ the degree of $v$ is at most 5. There are 7 cases:
			\begin{enumerate}
				\item $v\in \{\ell_{k+1},r_{k+1}\}$: Note that $\ell_{k+1}$ is a root of a binary tree, and other than nodes in $T(\ell_{k+1})$, it is connected only to the nodes $\ell_k$ and $r_{k+1}$, thus the degree of $\ell_{k+1}$ is 4. A similar argument holds for $r_{k+1}$.
				\item $v\in \{\ell_{k},r_{k}\}$: Note that $\ell_{k}$ is a root of a binary tree, and other than nodes in $T(\ell_{k})$, it is connected only to the node $\ell_{k+1}$, thus, the degree of $\ell_{k}$ is 3. A similar argument holds for $r_{k}$ as well.
				\item $v\in L\cup R$: Note that for each $\ell_i\in L$, it holds that $\ell_i$ is a root of a binary tree, and it is a leaf in the binary tree rooted at $x_1$. In addition, it is connected to $q(\ell_i)$ and to another node connecting it to a path of length $\log\log(k)-1$, which is connected to the binary tree rooted at $\ell_k$. Thus, its degree is 5. The same holds for each $r_i\in R$.
				\item $v\in P(x_1,x_2)$: If $v\in P(x_1,x_2)\setminus \{x_1\}$ then $v$ is of degree 2. Otherwise, the degree of $x_1$ is 3, since it is a root of a binary tree and it is connected to another node on the path $P(x_1,x_2)$.
				\item $v\in \emph{bit-gadget}$: Note that $v$ is a root of a binary tree and it is connected to two additional nodes in the \emph{bit-gadget}. Thus, its degree is 4.
				\item $v$ is an inner node in some binary tree: The degree of an inner node in a binary tree is at most 3.
				\item $v\in P(\ell_k)\cup P(r_k)\cup P(\ell_{k+1})\cup P(r_{k+1})\setminus (L\cup R)$: All the nodes on these paths are of degree 2.
			\end{enumerate}
		\end{proof}
	\end{claim}
	
	\subsection{Reduction from Set Disjointness}
	
	Each player receives an input string ($S_a$ and $S_b$) of $k$ bits. If $S_a[i]=0$, Alice removes the edge connecting $\ell_i$ to $q(\ell_i)$. Similarly, if $S_b[i]=0$, Bob removes the edge connecting $r_i$ to $q(r_i)$.
	
	\begin{claim}\label{deg3FirstDirection}
		For every $0\leq i,j\leq k-1$ such that $i\neq j$, it holds that $d(\ell_i,r_j)\leq 2\log\log(k)+2\log(k) - 1$.
		\begin{proof}
			If $i\neq j$, there must be some bit $h$, such that $i^h\neq j^h$. Assume without loss of generality that $i^h=1$ and $j^h=0$. Then, the distance between $\ell_i$ and $t_h$ is $\log\log(k)+\log(k)-1$. Similarly, the distance between $r_i$ and $t'_h$ is $\log\log(k)+\log(k)-1$. Note that $f_h$ and $t'_h$ are connected by an edge. Thus, $d(\ell_i,r_j)\leq 2\log\log(k)+2\log(k) + 2(- 1)+1=2\log\log(k)+2\log(k) - 1$.
		\end{proof}
	\end{claim}
	
	\begin{claim}\label{deg3SecondDirection}
		If $0\leq i\leq k-1$ is such that $S_a[i]=0$ or $S_b[i]=0$, then $d(\ell_i,r_i)\geq 2\log\log(k)+2\log(k)$. Otherwise, $d(\ell_i,r_i)=2\log\log(k)+2\log(k) - 1$.
		\begin{proof}
			If $0\leq i\leq k-1$ is such that $S_a[i]=0$ or $S_b[i]=0$, then either $\ell_i$ is not connect to $q(\ell_i)$ directly by an edge, or $r_i$ is not connect to $q(\ell_i)$ directly by an edge. Thus, similar to the previous constructions there are two options for any shortest path between $\ell_i$ and $r_i$. The first is through the \emph{bit-gadget} and the second is through the nodes $\ell_{k+1},r_{k+1}$, and both of them must have length at least $2\log\log(k)+2\log(k)$. Otherwise, in case $S_a[i]=1$ and $S_b[i]=1$ it holds that $d(\ell_i,\ell_{k+1})=\log\log(k)+\log(k)-1$. Similarly, $d(r_i,r_{k+1})=\log\log(k)+\log(k)-1$. Thus, $d(\ell_i,r_i)=2\log\log(k)+2\log(k)-1$.
		\end{proof}
	\end{claim}
	
	\begin{claim}\label{e(li)}
		If $0\leq i\leq k-1$ is such that $S_a[i]=1$ and $S_b[i]=1$, then $e(\ell_i)= 2\log\log(k)+2\log(k)-1$.
		\begin{proof}
			We show that for any $v\in V$, it holds that $d(\ell_i,v)\leq 2\log\log(k)+2\log(k)-1$. There are 2 cases:
			\begin{enumerate}
				\item $v\in V_a$: Here, there are 7 sub-cases:
				\begin{enumerate}
					\item $v\in \{\ell_k,\ell_{k+1}\}$: By the construction, $d(\ell_i,\ell_k)=\log\log(k)+\log(k)-1$, and since $S_a[i]=1$, it holds that $d(\ell_i,\ell_{k+1})=\log\log(k)+\log(k)-1$ as well.
					\item $v\in L$: In this case, $\ell_i$ can use the node $\ell_k$ in order to reach any node in $L$ by a path of length $2\log\log(k)+2\log(k)-2$, since the distance between any $\ell_i$ and $\ell_k$ is $\log\log(k)+\log(k)-1$.
					\item $v\in P(x_1,x_2)\cup T(x_1)\setminus L$: Note that $d(\ell_i,x_1)=\log(k)$, and $x_1$ can reach any node in $T(x_1)$ in $\log(k)$ steps, and any node in $P(x_1,x_2)$ in $\log(k)+2\log\log(k)-1$ steps. Thus, $d(\ell_i,v)\leq 2\log\log(k)+2\log(k)-1$.
					\item $v\in F\cup T$: Note that for each $0\leq j \leq \log(k)-1$, it holds that the distance from $\ell_i$ to one of the nodes in $\{f_j,t_j\}$ is $\log\log(k)+\log(k)-1$, and since we keep the edges between the nodes $f_j\in F$ and $t_j\in T$ for each $0\leq j \leq \log(k)-1$, it holds that $d(\ell_i,v)\leq \log(k)+\log\log(k)$.
					\item $v\in \bigcup_{i\in [k-1]} T(\ell_i)\bigcup_{j\in [\log(k)-1]} T(f_j)\bigcup_{j\in [\log(k)-1]} T(t_j)$: Note that $v$ is at distance at most $\log(k)+\log\log(k)-2$ from some node in $F\cup T$. In this case $d(\ell_i,v)\leq 2\log(k)+2\log\log(k)-2$ by case 1(d).
					\item $v\in T(\ell_k)\cup P(\ell_k)\setminus L$: Note that $v$ is at distance at most $\log(k)+\log\log(k)-1$ from $\ell_k$. In this case, $d(\ell_i,v)\leq 2\log(k)+2\log\log(k)-2$ by case 1(a).
					\item $v\in T(\ell_{k+1})\cup P(\ell_{k+1})\setminus L$: Note that $v$ is at distance at most $\log(k)+\log\log(k)-1$ from $\ell_{k+1}$. In this case, $d(\ell_i,v)\leq 2\log(k)+2\log\log(k)-2$ by case 1(a).
				\end{enumerate}
				\item $v\in V_b$: Here, there are 6 cases:
				\begin{enumerate}
					\item $v\in \{r_k,r_{k+1}\}$: Since $S_a[i]=1$, $d(\ell_i,r_{k+1})=d(\ell_i,\ell_{k+1})+d(\ell_{k+1},r_{k+1})=\log\log(k)+\log(k)$, and $d(\ell_i,r_{k})=\log\log(k)+\log(k)+1$.
					\item $v\in R$: By Claims ~\ref{deg3FirstDirection} and ~\ref{deg3SecondDirection}, it holds that $d(\ell_i,v)=2\log\log(k)+2\log(k)-1$.
					\item $v\in F'\cup T'$: By case 1(d), it holds that $d(\ell_i,v)\leq \log(k)+\log\log(k)+1$.
					\item $v\in \bigcup_{i\in [k-1]} T(r_i)\bigcup_{j\in [\log(k)-1]} T(f'_j)\bigcup_{j\in [\log(k)-1]} T(t'_j)$: Note that $v$ is at distance at most $\log(k)+\log\log(k)-2$ from some node in $F'\cup T'$. In this case $d(\ell_i,v)\leq 2\log(k)+2\log\log(k)-1$ by case 2(c).
					\item $v\in T(\ell_k)\cup P(r_k)\setminus L$: Note that $v$ is at distance at most $\log(k)+\log\log(k)-1$ from $r_k$. In this case $d(\ell_i,v)\leq 2\log(k)+2\log\log(k)-1$ by case 2(a).
					\item $v\in T(\ell_{k+1})\cup P(r_{k+1})\setminus L$: Note that $v$ is at distance at most $\log(k)+\log\log(k)-1$ from $r_{k+1}$. In this case $d(\ell_i,v)\leq 2\log(k)+2\log\log(k)-1$ by case 2(a).
				\end{enumerate}
			\end{enumerate}
			
		\end{proof}
	\end{claim}

	\begin{lemma}\label{lemma:mainDeg}
		If the two sets of Alice and Bob are disjoint, then the radius of $G$ is at least $2\log\log(k)+2\log(k)$. Otherwise, there is some $0\leq i\leq k-1$ such that $S_a[i]=1$ and $S_b[i]=1$ and the radius of $G$ is at most $e(\ell_i)\leq 2\log\log(k)+2\log(k)-1$.
		\begin{proof}
			Consider the case in which the two sets are not disjoint. By Claim ~\ref{deg3SecondDirection}, for all the nodes in $L$, it holds that $d(\ell_i,r_i)\geq 2\log\log(k)+2\log(k)$. Note that for all the nodes $u\in V\setminus (L\cup T(x_1)\cup P(x_1,x_2))$, it holds that $d(u,x_2)\geq 2\log\log(k)+2\log(k)$, and for all the nodes $u\in T(x_1)\cup P(x_1,x_2)$ it holds that $d(u,r_i)\geq 2\log\log(k)+2\log(k)$ for each $0\leq i\leq k-1$. Thus, the radius of $G$ is at least $2\log\log(k)+2\log(k)$ as well. The second part of the claim follows directly from Claim ~\ref{e(li)}.
		\end{proof}
	\end{lemma}

	\paragraph{Proof of Theorem ~\ref{ExactRadsmallDeg}} As described before, we add $O(k)$ binary trees of size $O(\log(k))$, and $O(\log(k))$ binary trees of size $O(k)$, and $O(k)$ paths of size $O(\log\log(k))$. Thus the total number of nodes added to the construction described in Section ~\ref{ExactRad} in $O(k\log(k))$. Thus, $k=\Omega(\frac{n}{\log(n)})$ and by Lemma ~\ref{lemma:mainDeg} that any algorithm for computing the radius of $G$ requires at least $\Omega(\frac{n}{\log^3(n)})$ rounds, even in graphs with $\Delta=5$.
	
	\begin{remark}
		We remark that we believe that by a simple modification we can obtain the same result for graphs with maximum degree 3. This is done by replacing each node of degree 4 by 2 nodes, each of degree 3, and by replacing each node of degree 5 by 3 nodes, each with degree 3.
	\end{remark}

	\section{Verification of Spanners}\label{spanners}
	
	In this section we show that verifying an $(\alpha,\beta)$-spanner is also a hard task in the $CONGEST$ model.
	
	\begin{theorem-repeat}{thm:spanners}
		\ThmSpanners
	\end{theorem-repeat}

	In particular, this gives a near-linear lower bound when $\alpha,\beta$ are at most polylogarithmic in $n$.
	
	\subsection{Graph Construction}
	
	\begin{figure}[h]
		\begin{center}
			\includegraphics[scale=0.6]{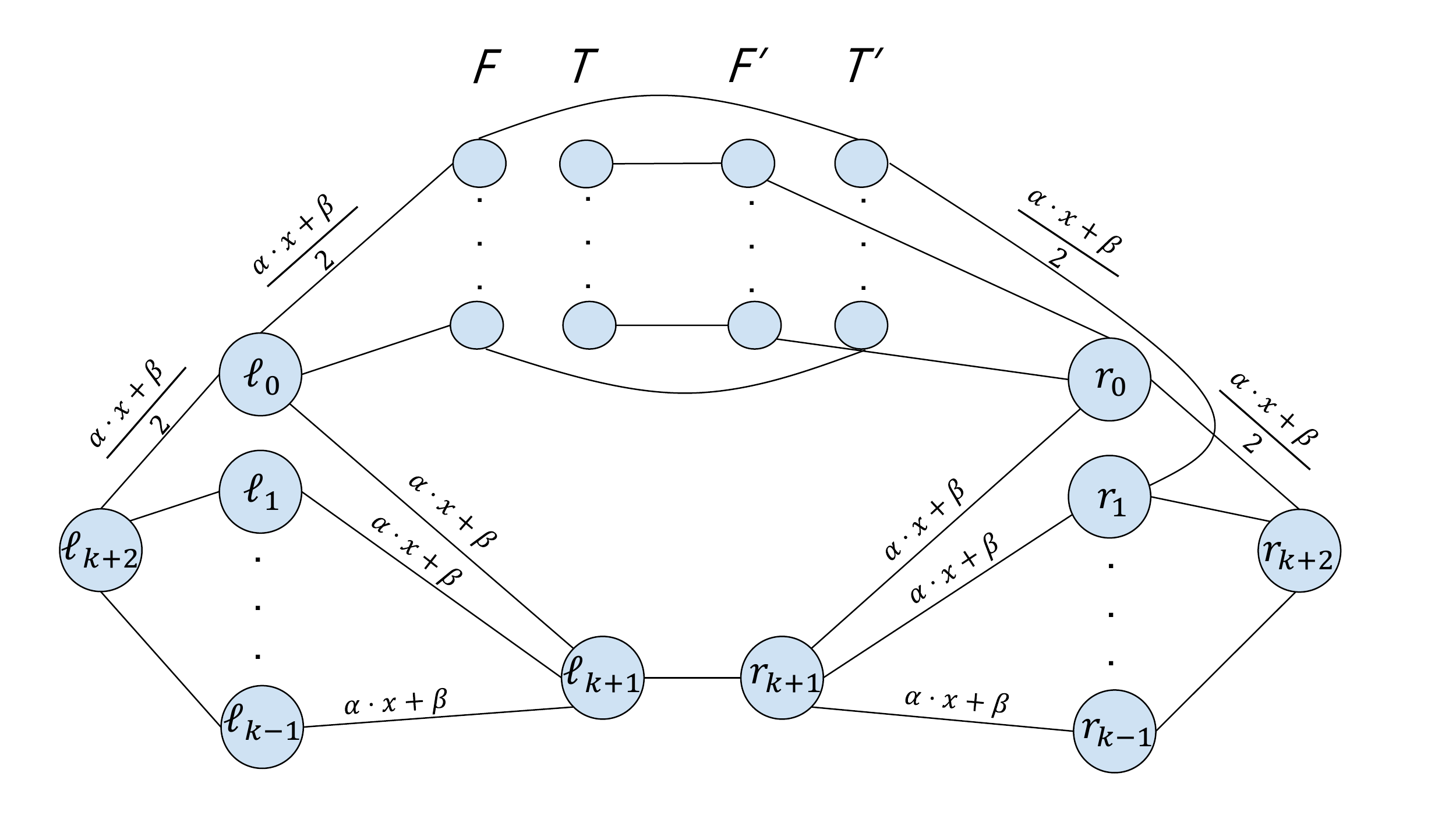}
		\end{center}
		\caption{Graph Construction (spanner verification).}
	\end{figure}

	Let $P=\alpha x +\beta$. For our proof for \emph{unweighted} graphs we need $x$ to be equal to 1. However, we keep the notation $x$ to prove the same lower bound for \emph{weighted} graphs as well.
	We apply the following changes to the construction described in Section~\ref{rad} (see also Figure 8):\footnote{Some of the edges are omitted, for clarity.}
	
	\begin{enumerate}
		\item Remove the nodes $\{x_1,x_2,x_3\}$ and their incident edges.
		\item Remove the nodes $\{\ell_k,r_k\}$ and their incident edges.
		\item Remove all the edges $(u,v)$ such that $u\in F$ and $v\in T$. Similarly, Remove all the edges $(u,v)$ such that $u\in F'$ and $v\in T'$.
		\item Replace all the edges $(u,v)$ such that $u\in L$ and $v\in (F\cup T)$ by paths of length $\frac{P}{2}$. Similarly, Replace all the edges $(u,v)$ such that $u\in R$ and $v\in (F'\cup T')$ by paths of length $\frac{P}{2}$.
		\item Connect each $\ell_i\in L$ to $\ell_{k+1}$ by a path of length $P$. Similarly, Connect each $r_i\in R$ to $r_{k+1}$ by a path of length $P$.
		\item Add two additional nodes $\{\ell_{k+2},r_{k+2}\}$. Connect $\ell_{k+2}$ to each of the nodes in $L$ by a path of length $\frac{P}{2}$. Similarly, connect $r_{k+2}$ to each of the nodes in $R$ by a path of length $\frac{P}{2}$.
	\end{enumerate}
	
	The following Observation follows directly from the construction and the discussions in previous sections.
	
	\begin{observation} \label{spannersOb}
		For every $0\leq i,j\leq k-1$ such that $i\neq j$, it holds that $d(\ell_i,r_j)=\alpha x +\beta + 1$. Otherwise, $d(\ell_i,r_i)=2(\alpha x +\beta) + 1$.
	\end{observation}
	
	\subsection{Reduction from Set-Disjointness}
	
	Each player receives an input string ($S_a$ and $S_b$) of $k$ bits. If $S_a[i]=1$, Alice adds a path of length $x$ between the nodes $\ell_i$ and $\ell_{k+1}$. Similarly, if $S_b[i]=1$, Bob adds a path of length $x$ between the nodes $r_i$ and $r_{k+1}$.
	Denote by $Dist$ the set of all the edges that are added according to the strings of Alice and Bob.
	Now we define the subgraph $H$ to be $E\setminus Dist$. Let $d_G(u,v)$ denote the distance between the nodes $u$ and $v$ in $G$. Similarly, let $d_H(u,v)$ denote the distance between the nodes $u$ and $v$ in $H$.
	
	\begin{lemma}\label{Firstdirection} If there is some $0\leq i\leq k-1$ such that $S_a[i]=1$ and $S_b[i]=1$, then $H$ is not an $(\alpha,\beta)$ spanner of $G$ for any $\alpha<\beta+1$.
		\begin{proof}
			According to Observation~\ref{spannersOb}, it holds that $d_H(\ell_i,r_i)=2(\alpha x +\beta) + 1$, while $d_G(\ell_i,r_i)=2x + 1$. And since $\alpha (2 x +1)+\beta < 2(\alpha x +\beta) + 1$ for any $\alpha<\beta+1$, it holds that $H$ is not an $(\alpha,\beta)$ spanner of $G$.
		\end{proof}
	\end{lemma}
	
	\begin{lemma}\label{Seconddirection}
		If the two sets of Alice and Bob are disjoint, then $H$ is an $(\alpha,\beta)$ spanner of $G$ for any $\alpha\geq 1$.
		\begin{proof}
			We show that $d_H(u,v)\leq \alpha d_G(u,v) + \beta$ for any $u,v\in V$.
			\begin{enumerate}
				\item $u=\ell_i\in L$ and $v=r_j\in R$ such that $i\neq j$: Note that $d_G(\ell_i,r_j)= \min(2 x+1, \alpha x +\beta + 1)$, while $d_H(\ell_i,r_j)=\alpha x +\beta + 1$ by Observation~\ref{spannersOb}. Thus, the problematic case is when $d_G(\ell_i,r_j)=2x+1$, for which $\alpha(d_G(\ell_i,r_j))+\beta\geq\alpha (2x+1)+\beta=2 \alpha x + \alpha + \beta\geq \alpha x + \beta + 1=d_H(\ell_i,r_j)$ for any $\alpha\geq 1$.
				\item $u=\ell_i\in L$ and $v=r_j\in R$ such that $i=j$: Note that $d_G(\ell_i,r_i)\geq x+\alpha x + \beta + 1$, since either $\ell_i$ is not connected to $\ell_{k+1}$ by a path of length $x$, or $r_i$ is not connected to $r_{k+1}$ by a path of length $x$. While $d_H(\ell_i,r_i)=2(\alpha x +\beta) + 1$ by Observation~\ref{spannersOb}. Note that $\alpha(d_G(\ell_i,r_i))+\beta\geq\alpha(x+\alpha x + \beta + 1)+\beta\geq \alpha x+\alpha^2 x+\alpha \beta +\alpha+\beta\geq \alpha x+\alpha x+\beta+1+\beta\geq 2(\alpha x+\beta)+1=d_H(\ell_i,r_i)$, for any $\alpha\geq 1$.
				\item $u\in V_a$ and $v\in V_a$: Here, there are two cases:
				\begin{enumerate}
					\item The shortest path in $G$ does not visit the node $\ell_{k+1}$. In this case $d_G(u,v)=d_H(u,v)$.
					\item The shortest path in $G$ visits the node $\ell_{k+1}$, thus, it visits some node $\ell_i$. In this case $d_G(u,v)$ can be written as the sum of three distances, $d_G(u,v)=d_G(u,\ell_i)+d_G(\ell_i,\ell_{k+1})+d_G(\ell_{k+1},v)$, if none of the distances $d_G(u,\ell_i),d_G(\ell_{k+1},v)$ is using the disjointness edges, then the claim holds by the fact that $\alpha (d_G(\ell_i,\ell_{k+1}))+\beta\geq d_H(\ell_i,\ell_{k+1})$. Otherwise, $d_G(u,v)$ can be written as the sum of three distances, $d_G(u,v)=d_G(u,\ell_i)+d_G(\ell_i,\ell_j)+d_G(\ell_j,v)$, note that $d_G(\ell_i,\ell_j)\geq 2x$, while $d_H(\ell_i,\ell_j)=\alpha x+\beta$ (the path through the node $\ell_{k+2}$), and the claim holds by the fact that $\alpha 2x+\beta\geq \alpha x+\beta$.
				\end{enumerate}
				\item $u\in V_b$ and $v\in V_b$: This case is symmetric to the previous one.
				\item $u\in V_a$ and $v\in V_b$: Here, there are four cases:
				\begin{enumerate}
					\item The shortest path in $G$ does not visit the nodes $\ell_{k+1},r_{k+1}$. In this case $d_G(u,v)=d_H(u,v)$.
					\item The shortest path in $G$ visits only the node $\ell_{k+1}$ and does not visit $r_{k+1}$: This case is very similar to the case 3.b.
					\item The shortest path in $G$ visits only the node $r_{k+1}$ and does not visit $\ell_{k+1}$: This case is symmetric to the previous case.
					\item The shortest path visits the two nodes $\ell_{k+1},r_{k+1}$, thus, it visits some two nodes $\ell_i,r_j$. In this case $d_G(u,v)$ can be written as the sum of three distances, $d_G(u,v)=d_G(u,\ell_i)+d_G(\ell_i,r_j)+d_G(r_j,v)$, in case $i\neq j$, it holds that $\alpha(d_G(u,v))+\beta\geq d_H(u,v)$ by the case in which $u=\ell_i\in L$ and $v=r_j\in R$ such that $i\neq j$ which was proved in the first case of this proof. Otherwise, $\alpha(d_G(u,v))+\beta\geq d_H(u,v)$ by the case in which $u=\ell_i\in L$ and $v=r_j\in R$ such that $i=j$ which was proved in the second case of this proof.
				\end{enumerate}
			\end{enumerate}
		\end{proof}
	\end{lemma}
	
	\begin{observation}
		\label{observation:SpannerSparsity}
		Note that the number of edges $E(G)$ is equal to $\Theta(E(H))$. However, it is straightforward to see that by adding a clique of size $k$ to $G$ and connecting it to some arbitrary node we can control the number of edges in $G$.  We add $k$ additional edges to $H$ in order to span this clique, giving that the number of edges in $G$ is $\Theta(n^2)$, while the number of edges in $H$ is $\Theta(n\log(n))$ (and if $\beta>2$ the number of edges in $H$ is equal to $\Theta(n)$). If we want $E(H)$ to match the known bounds for the possible sparsity of an $(\alpha,\beta)$-spanner for general graphs, we simply add more edges of the clique to $H$.
	\end{observation}
	
	\paragraph{Proof of Theorem~\ref{thm:spanners}} From the two lemmas~\ref{Firstdirection} and~\ref{Seconddirection} we deduce that $H$ is an $(\alpha,\beta)$-spanner of $G$ if and only if the two sets of Alice and Bob are disjoint. Note that for \emph{unweighted} graphs we need $x$ to be equal to 1, thus, $n=O(k\cdot (\alpha+\beta) \log(n))$, i.e., $k=\Omega(\frac{n}{(\alpha+\beta) \log(n)})$, and as in the previous constructions from the previous sections, the size of the cut is $O(\log(n))$. Thus, the number of rounds needed for any algorithm to decide whether $H$ is an $(\alpha,\beta)$-spanner of $G$ is $\Omega(\frac{n}{(\alpha+\beta)\cdot \log^3(n)})$.
	
	\paragraph{\emph{Weighted Graphs}}
	In order to achieve a higher lower bound for $\emph{weighted}$ graphs, we replace all the $P$ and $\frac{P}{2}$ paths by edges of weights $P$ and $\frac{P}{2}$ respectively. Note that in this case $k=\Omega(n)$ and we deduce that the number of rounds needed for any algorithm to decide whether $H$ is an $(\alpha,\beta)$-spanner of $G$ is $\Omega(\frac{n}{\log^2(n)})$.

\section{Discussion}

We introduce a new technique for reducing the Set-Disjointness communication problem to distributed computation problems, in a highly efficient way.
Our reductions encode an instance of Disjointness on $k$ bits into a graph on only $\widetilde{O}(k)$ nodes and edges with a small ``communication-cut" of size $O(\log{k})$.
All previous lower bound constructions had a cut of $poly(k)$ size (e.g.,~\cite{HolzerP14,HolzerW12,SarmaHKKNPPW12,FrischknechtHW12}).
This efficiency allows us to answer several central open questions regarding the round complexity of important distance computation problems in the $CONGEST$ model.

There are several interesting directions for future work.
First, there is still a $\log{n}$ factor gap between the upper and lower bounds on the round complexity of computing the diameter in the $CONGEST$ model.
Due to the fundamentality of the diameter, we believe that it will be interesting to close this small gap.

Second, we believe our degree-reduction technique can be used also for our bounds on approximations, albeit with more involved modifications. We also believe this technique can be useful in obtaining lower bounds on constant degree graphs in additional settings, beyond CONGEST.

Third, while our ideas greatly improve the state of the art lower bounds for shortest paths problems on \emph{unweighted} graphs, their potential in the regime of \emph{weighted} graphs is yet to be explored.

Finally, following our strong barriers for sparse graphs, it is important to seek further natural restrictions on the networks that would allow for much faster distance computation.
Planar graphs are an intriguing setting in this context.
A promising recent work of Ghaffari and Haeupler~\cite{GhaffariH16} showed that computing a minimum spanning tree can be done in $\widetilde{O}(D)$ rounds in planar graphs, despite the $\widetilde{\Omega}(\sqrt{n})$ lower bound for general graphs~\cite{SarmaHKKNPPW12}. Can the diameter of a planar network be computed in $\widetilde{O}(D)$ rounds?
While the graphs in our lower bounds are highly non-planar, it is interesting to note that they have a relatively small treewidth of $O(\log{n})$.

\paragraph{Acknowledgement:} \hspace{-0.3cm} We thank Ami Paz for pointing out Observation~\ref{observation:SpannerSparsity}, and for useful discussions.

\bibliographystyle{abbrv}
\bibliography{sigproc}

\end{document}